\documentclass[11pt]{article}

\usepackage{fullpage}
\usepackage{soul}
\usepackage{url}
\usepackage[utf8]{inputenc}
\usepackage[small]{caption}
\usepackage{graphicx}
\usepackage{xcolor}
\usepackage{amsmath}
\usepackage{mathtools}
\mathtoolsset{showonlyrefs = true}
\usepackage{booktabs}
\usepackage{algorithm}
\usepackage{enumerate}
\usepackage{dsfont}
\usepackage[hidelinks]{hyperref}
\usepackage{cleveref}
\usepackage{comment}

\newcommand{\du}{\mathrm{disutil}}
\newcommand{\mechWF}{\varphi^{\mathrm{WF}}}
\newcommand{\mechQP}{\varphi^{\mathrm{QP}}}
\DeclareMathOperator{\EG}{EG}

\usepackage{amsthm}
\usepackage{amssymb}
\usepackage{bm}
\usepackage{physics}

\newtheorem{theorem}{Theorem}
\newtheorem{corollary}{Corollary}
\newtheorem{proposition}{Proposition}
\newtheorem{lemma}{Lemma}

\theoremstyle{definition}
\newtheorem{definition}{Definition}
\newtheorem{example}{Example}

\usepackage{natbib}
\usepackage{authblk}

\allowdisplaybreaks

\title{\bf Fractional Assignment with $\ell_1$ Preferences}

\author[1]{Yasushi Kawase}
\author[2]{Warut Suksompong}
\author[3]{Hanna Sumita}
\author[3]{Yu Yokoi}

\affil[1]{Chuo University, Japan}
\affil[2]{National University of Singapore, Singapore}
\affil[3]{Institute of Science Tokyo, Japan}

\date{\vspace{-10mm}}

\begin{document}

\maketitle

\begin{abstract}
We study a fractional assignment setting where $n$ objects are to be assigned to $n$ agents with unit capacity, and each agent specifies an ideal distribution over the objects.
Unlike in classic random assignment, these ideal distributions are not necessarily degenerate, as agents may prefer a mixture of objects rather than any single object.
We assume that agents seek to minimize the $\ell_1$ distance between their ideal distribution and the distribution they receive, which is equivalent to maximizing the overlap between the two distributions.
We propose two mechanisms, one based on water filling (WF) and the other on quadratic programming (QP), and show that both mechanisms are utilitarian-optimal (and hence Pareto efficient), envy-free, strategyproof, and satisfy equal treatment of equals.
Moreover, we highlight a distinct advantage of each mechanism: while the WF mechanism satisfies the stronger property of group-strategyproofness, the QP mechanism is more robust in terms of egalitarian overlap welfare.
\end{abstract}

\section{Introduction}
\label{sec:intro}

In the classic \emph{random assignment} problem, there are $n$ agents (e.g., students, workers, or tenants) and $n$ objects (e.g., school seats, job positions, or public houses), and the goal is to compute a random assignment of the objects to the agents based on their preferences.
Such an assignment is described by a doubly stochastic matrix, where each agent's probabilities of being assigned the $n$~objects sum to~$1$, and each object's probabilities of being assigned to the $n$~agents also sum to~$1$.
Agents are assumed to have preferences in the form of \emph{ordinal rankings} over the objects---thus, each agent's most preferred outcome is to be assigned her most preferred object with probability~$1$.
Since the pioneering work of \citet{BogomolnaiaMo01}, random assignment has been extensively studied in both economics and computer science \citep{AbdulkadirogluSo03,KattaSe06,AzizMaXi15,BrandlBrHo17,KawaseSuYo23,Duddy25}.

In this paper, we consider a variant of random assignment recently introduced by \citet{Ramezanian24}, where agents do not have ordinal preferences, and therefore do not necessarily prefer degenerate distributions.
Instead, each agent specifies her \emph{ideal distribution} (or \emph{target distribution}) over objects, which may be non-degenerate.
Such preferences can arise in several practical settings.
For instance, when the objects are divisible---such as time for using different facilities (e.g., gym, laboratory, and study space) or completing various tasks (e.g., research, teaching, and paper reviewing)---agents may derive greater benefit from splitting their time across multiple facilities or tasks rather than concentrating it all on a single option.
Another example is a national government allocating resources across domains such as education, healthcare, and transportation to its cities---it is likely that each city would prefer a diversified bundle spanning multiple resource types instead of receiving all of its resources in a single domain.
Due to these interpretations of distributions as (deterministic) assignments of divisible objects rather than as lotteries over indivisible objects, we shall refer to our problem as \emph{fractional assignment}.
Nevertheless, our problem can also be applied when indivisible objects are assigned repeatedly over time, such as household chores across days, as agents may prefer rotating over a mix of chores rather than handling the same chore every day. A similar motivation based on repeated matching scenarios was also discussed by \citet{Ramezanian24}.

In order to formally investigate our fractional assignment problem, we need to specify how agents evaluate possible outcomes other than their target distributions.
A natural approach, also taken by \citet{Ramezanian24}, is to consider an agent's disutility to be the ``distance'' between her target distribution and the distribution that she receives.
While several distance measures are known in the literature, a particularly well-studied one is the \emph{$\ell_1$ distance}, defined as the difference between the two distributions on each object, summed up across all objects.
Besides its simplicity, minimizing the $\ell_1$ distance is equivalent to maximizing the \emph{overlap} between the two distributions, i.e., the minimum between the two distributions on each object, summed up across all objects.\footnote{See, e.g., Lemma~2.6 of \citet{GoelKrSa19}.}
For these reasons, $\ell_1$ preferences are commonly assumed in work on \emph{budget aggregation}, a ``public-goods'' version of our problem where the agents' target distributions are aggregated into a single output distribution shared by all agents \citep{FreemanPePe21,CaragiannisChPr24,deBergFrSc24,BrandtGrSe26,CembranoFrSc26,ElkindGrLe26}.\footnote{See also the broader survey by \citet{SuksompongTe26}.}

With preferences defined based on $\ell_1$ distances, we can now formalize the notions that we desire in the fractional assignment setting.
Three fundamental desiderata in preference aggregation problems are efficiency, fairness, and strategyproofness.
An assignment is \emph{Pareto efficient} if no other assignment makes some agent better off and no agent worse off.
Fairness can be captured by \emph{envy-freeness}---an assignment is envy-free if no agent prefers another agent's assignment to her own---as well as by \emph{equal treatment of equals}, which stipulates that agents who have the same preference should receive the same outcome.
A mechanism is said to be \emph{strategyproof} if no agent can benefit by misreporting her target distribution, and Pareto efficient (resp., envy-free) if it always produces a Pareto efficient (resp., envy-free) assignment.
Does there exist a fractional assignment mechanism that is Pareto efficient, envy-free, strategyproof, and satisfies equal treatment of equals under $\ell_1$ preferences?

\subsection{Overview of Results}

We answer the question above in the affirmative.
In particular, we propose two natural mechanisms for fractional assignment with $\ell_1$ preferences: the \emph{water filling (WF)} mechanism and the \emph{quadratic programming (QP)} mechanism.
\begin{itemize}
\item The WF mechanism constructs the output allocation in an iterative manner.
In the first phase, it fills entries at the same speed, stopping an entry when either the corresponding cap is reached or the object is exhausted.
In the second phase, it continues this filling process without the cap condition until all agents are filled and all objects are exhausted.\footnote{The iterative filling nature bears a resemblance to the ``probabilistic serial'' mechanism in random assignment with ordinal preferences \citep{BogomolnaiaMo01}.}
\item By contrast, the QP mechanism operates in a single step, returning the ``minimum $\ell_2$-norm'' allocation (i.e., the allocation having the least sum of squared entries) within the set of utilitarian-optimal allocations.\footnote{A similar tie-breaking rule has been used in budget aggregation \citep{Lindner11,FreemanPePe21}.}
\end{itemize}
We prove that both mechanisms are utilitarian-optimal, envy-free, strategyproof, and satisfies equal treatment of equals.

Furthermore, we demonstrate that the two mechanisms have differing strengths.
On the one hand, the WF mechanism satisfies the group-strategyproofness guarantee, 
whereas the QP mechanism fails even weak group-strategyproofness. 
On the other hand, the QP mechanism is more robust in terms of ``egalitarian overlap welfare'' (i.e., the minimum overlap utility among all agents).
Specifically, while the egalitarian overlap welfare produced by the QP mechanism can exceed that of the WF mechanism by a factor of $\Omega(n)$, 
it can never fall short by more than a factor of $2$.

\subsection{Related Work}

The prior work most closely related to ours is the aforementioned work of \citet{Ramezanian24}, who studied an equivalent setting under the name ``division of chances''.
\citeauthor{Ramezanian24} proposed a class of mechanisms called \emph{uniform rule for dividing chances (URC)}, which are also Pareto efficient, envy-free, and strategyproof.
While the first phase of URC mechanisms works in a similar way to our WF mechanism, their second phase operates sequentially rather than simultaneously.
As a consequence, URC mechanisms violate equal treatment of equals, which WF satisfies.
\citeauthor{Ramezanian24} did not consider an optimization approach that our QP mechanism takes, nor the comparison based on egalitarian welfare.
In fact, our results answer three open questions posed by \citeauthor{Ramezanian24}---for more details, see Appendix~\ref{app:Ramezanian}.
Earlier work focused on the \emph{uniform rule} for dividing resources among agents with single-peaked preferences \citep{Sprumont91,MorimotoSeCh13}.
As discussed by \citet[p.~5]{Ramezanian24}, this method is unsuitable for the fractional assignment setting.

The allocation of multiple homogeneous divisible objects has also been studied by \citet{GhodsiZaHi11} and \citet{AzizHeLu25}.
\citet{GhodsiZaHi11} assumed that the agents' preferences are based on \emph{Leontief utilities} drawn from economics: if agent~$i$ needs $p_{ij}$ units of object~$j$ for each production unit and receives $x_{ij}$ units of object~$j$, then her utility is the minimum among the ratios $x_{ij}/p_{ij}$ among all~$j$.
These authors proposed a solution called ``dominant resource fairness'', which satisfies several desirable properties including envy-freeness, strategyproofness, and Pareto efficiency.
\citet{AzizHeLu25} considered a model where agents may have non-linear utilities within each object, but these utilities are additive across objects.
While their model is more general than ours, the generality makes strong guarantees difficult to obtain.
In particular, they showed that deciding the existence of an envy-free and Pareto efficient outcome is NP-complete, and no sublinear approximation of a fairness notion called ``maximin share fairness'' can be satisfied.

Further afield, the problem of allocating a divisible resource is often referred to as \emph{cake cutting}, where the cake is a metaphor for a \emph{heterogeneous} divisible resource such as land \citep{Procaccia16}.
Since the utilities in cake cutting can be complex and difficult to describe explicitly, a common method for an algorithm to access them is via queries \citep{RobertsonWe98}.
While different agents may generally receive different amounts of cake in this model, variants in which all agents must receive the same amount have also been studied \citep{JojicPaZi21}.

\section{Preliminaries}
\label{sec:model}

\subsection{Model}
\label{subsec:model}

A fractional assignment instance consists of a tuple $(N,M,(\bm{p}_i)_{i\in N})$, where $N$ is a set of $n$
agents and $M$ is a set of $n$ objects, for some positive integer $n\ge 2$.
We assume that there is one unit of each object.
Let $\Delta$ denote the set of unit distributions over the objects in $M$, i.e.,
$\Delta\coloneqq\{\bm{q}\in\mathbb R^M\mid \sum_{j\in M}q_j=1,\ q_j\ge 0\ (j\in M)\}$.
Each agent $i\in N$ is associated with a target distribution $\bm{p}_i\in\Delta$.

A (fractional) allocation $X=(x_{ij})_{i\in N,\,j\in M}$ is an $n \times n$
doubly stochastic matrix, that is, its entries are nonnegative, each row sums to $1$, and
each column sums to $1$.
We denote by $\bm{x}_i=(x_{ij})_{j\in M}$ the row for agent $i\in N$.
In an allocation $X$, agent $i\in N$ receives a distribution $\bm{x}_i\in\Delta$ of
objects, i.e., an $x_{ij}$ fraction of each object $j\in M$.
Let $\mathcal{B}$ denote the set of all fractional allocations.
Given an allocation $X$, each agent $i\in N$ incurs a disutility due to the difference between $\bm{x}_i$ and $\bm{p}_i$.
We define the \emph{disutility} of agent $i\in N$ from $\bm{q}\in\Delta$ as the
$\ell_1$ distance between $\bm{q}$ and $\bm{p}_i$, that is,
$\du(\bm{q},\bm{p}_i)\coloneqq \sum_{j\in M}|q_j-p_{ij}|$.

An allocation $X=(x_{ij})_{i\in N,\,j\in M}$ is called \emph{utilitarian-optimal} if it
minimizes the total disutility of the agents, $\sum_{i\in N}\du(\bm{x}_i,\bm{p}_i)$.
It is \emph{Pareto efficient} if there is no allocation $Y=(y_{ij})_{i\in N,\,j\in M}\in\mathcal{B}$ such that
$\du(\bm{y}_i,\bm{p}_i)\le\du(\bm{x}_i,\bm{p}_i)$ for every $i\in N$, with a strict inequality for at least one
agent.
Moreover, an allocation $X$ is \emph{envy-free} if
$\du(\bm{x}_i,\bm{p}_i)\le \du(\bm{x}_{i'},\bm{p}_i)$ for all $i,i'\in N$.

A mechanism is a map $\psi\colon\Delta^N\to\mathcal{B}$ from a target distribution
profile to an allocation. A mechanism $\psi$ is \emph{strategyproof} if, for every
target distribution profile $P=(\bm{p}_i)_{i\in N}\in\Delta^N$, there is no pair of agent $i\in N$ and report
$\bm{q}_i\in\Delta$ such that
$\du(\psi(\bm{q}_i,\bm{p}_{-i})_i,\bm{p}_i)<\du(\psi(P)_i,\bm{p}_i)$.
It is \emph{weakly group-strategyproof} if, for every profile $P\in\Delta^N$, every nonempty
coalition $S\subseteq N$, and every joint misreport $\bm{q}_S\in\Delta^S$, not all members of $S$
strictly improve, that is, it is not the case that
$\du(\psi(\bm{q}_S,\bm{p}_{-S})_i,\bm{p}_i)<\du(\psi(P)_i,\bm{p}_i)$ for every $i\in S$.
It is \emph{group-strategyproof} if, under the same quantification, it is not the case that
$\du(\psi(\bm{q}_S,\bm{p}_{-S})_i,\bm{p}_i)\le\du(\psi(P)_i,\bm{p}_i)$ for every $i\in S$, with a strict inequality
for at least one $i\in S$.
A mechanism is \emph{Pareto efficient} (resp., \emph{envy-free}) if the allocation that it returns is always Pareto efficient (resp., envy-free).
Finally, a mechanism~$\psi$ satisfies \emph{equal treatment of equals} if for every profile $P=(\bm{p}_i)_{i\in N}$ and every pair of agents $i,j\in N$ such that $\bm{p}_i = \bm{p}_j$, it holds that $\psi(P)_i = \psi(P)_j$.

\subsection{Basic Properties of $\ell_1$ Preferences}
\label{subsec:basic-properties}

Since both $\bm{q}$ and $\bm{p}_i$ lie in $\Delta$, minimizing the $\ell_1$ disutility is equivalent to maximizing
the \emph{overlap utility} 
\begin{align}
u_i(\bm{q},\bm{p}_i)\coloneqq \sum_{j\in M}\min\{q_j,p_{ij}\}
=1-\frac12\,\du(\bm{q},\bm{p}_i).
\end{align}

We classify each object $j\in M$ by how its total reported demand compares with its supply of $1$:
$M_=\coloneqq\{j\in M\mid \sum_{i\in N}p_{ij}=1\}$,
$M_>\coloneqq\{j\in M\mid \sum_{i\in N}p_{ij}>1\}$, and
$M_<\coloneqq\{j\in M\mid \sum_{i\in N}p_{ij}<1\}$.
Define $M_\ge \coloneqq M_=\cup M_>$ and $M_\le \coloneqq M_=\cup M_<$.

The following characterization of utilitarian-optimal allocations is the common starting point for mechanisms studied in this paper.

\begin{lemma}\label{lem:uo-criterion}
An allocation $(x_{ij})_{i\in N,\,j\in M}$ is utilitarian-optimal if and only if, for every
$i\in N$ and $j\in M$, it holds that $x_{ij}=p_{ij}$ when $j\in M_=$,
$x_{ij}\le p_{ij}$ when $j\in M_>$, and $x_{ij}\ge p_{ij}$ when $j\in M_<$.
\end{lemma}

\begin{proof}
For any allocation $(x_{ij})_{i\in N,\,j\in M}$, the total disutility is given by
\begin{align}
\sum_{i\in N}\du(\bm{x}_i,\bm{p}_i)
&=\sum_{i\in N}\sum_{j\in M}|x_{ij}-p_{ij}|
=\sum_{j\in M}\sum_{i\in N}|x_{ij}-p_{ij}| \notag\\
&\ge \sum_{j\in M_>}\sum_{i\in N}(p_{ij}-x_{ij})
+\sum_{j\in M_<}\sum_{i\in N}(x_{ij}-p_{ij}) \label{eq:UO}\\
&=\sum_{j\in M_>}\left(\sum_{i\in N}p_{ij}-1\right)
+\sum_{j\in M_<}\left(1-\sum_{i\in N}p_{ij}\right). \notag
\end{align}
Equality in~\eqref{eq:UO} holds if and only if, for every $i\in N$, 
(i) $x_{ij}\le p_{ij}$ for each $j\in M_>$, 
(ii) $x_{ij}=p_{ij}$ for each $j\in M_=$, and
(iii) $x_{ij}\ge p_{ij}$ for each $j\in M_<$.

It remains to show that the lower bound is achievable. 
Define $c_j\coloneqq \sum_{i\in N}p_{ij}$ for each $j\in M$.
Then, 
(i) $c_j>1$ if $j\in M_>$, 
(ii) $c_j=1$ if $j\in M_=$, and
(iii) $c_j<1$ if $j\in M_<$.
Let $T\coloneqq \sum_{j\in M_>}(c_j-1) = \sum_{j\in M_<}(1-c_j) \ge 0$, and
$s_i\coloneqq \sum_{j\in M_>}(1-1/c_j)p_{ij}$ for each $i\in N$. 
Then,
$$\sum_{i\in N}s_i
= \sum_{i\in N}\sum_{j\in M_>}\left(1 - \frac{1}{c_j}\right)p_{ij}
= \sum_{j\in M_>}\sum_{i\in N}\left(\frac{c_j - 1}{c_j}\right)p_{ij}
= \sum_{j\in M_>}(c_j-1)=T.$$
If $T=0$, then $M_>=M_<=\emptyset$, and hence $(x_{ij})_{i\in N,\,j\in M}=(p_{ij})_{i\in N,\,j\in M}$ satisfies the condition.
Suppose now that $T>0$.
We define $x_{ij}$ for each $i\in N$ and $j\in M$ by
\begin{align*}
x_{ij}=
\begin{cases}
p_{ij}/c_j          & \text{if }j\in M_>,\\
p_{ij}              & \text{if }j\in M_=,\\
p_{ij}+(1-c_j)s_i/T & \text{if }j\in M_<.
\end{cases}
\end{align*}
Clearly, $x_{ij} \ge 0$.
For each $i\in N$, we have $\sum_{j\in M}x_{ij}=\sum_{j\in M}p_{ij}-\sum_{j\in M_>}(1-1/c_j)p_{ij}+\sum_{j\in M_<}(1-c_j)s_i/T=1-s_i+s_i=1$.
Moreover, the columns of $X$ have the following properties.
\begin{itemize}
\item For each $j\in M_>$, we have $x_{ij}\le p_{ij}$ and $\sum_{i\in N}x_{ij} = \sum_{i\in N}p_{ij}/c_j=c_j/c_j=1$.
\item For each $j\in M_=$, we have $x_{ij}=p_{ij}$ and $\sum_{i\in N}x_{ij} = \sum_{i\in N}p_{ij}=1$.
\item For each $j\in M_<$, we have $x_{ij}\ge p_{ij}$ and $\sum_{i\in N}x_{ij} = \sum_{i\in N}(p_{ij}+(1-c_j)s_i/T)=c_j+(1-c_j)\sum_{i\in N}s_i/T=c_j+(1-c_j)=1$.
\end{itemize}
Hence, the lower bound is attained, and the lemma follows.
\end{proof}

For each profile $P\in\Delta^N$, let $\mathcal{U}(P)$ denote the set of utilitarian-optimal allocations with respect to $P$.
By \Cref{lem:uo-criterion}, we have
\begin{align}
\mathcal{U}(P)=\left\{X\in\mathcal{B}\;\middle|\; \begin{array}{ll} 
x_{ij}\le p_{ij} & (i\in N,\ j\in M_\ge),\\
x_{ij}\ge p_{ij} & (i\in N,\ j\in M_\le)
\end{array}\right\}.
\label{eq:mathcalU}
\end{align}

\begin{example}\label{ex:small_instance}
Let $n=3$, $N = M = \{1,2,3\}$, $\bm{p}_1=(1,0,0)$, and $\bm{p}_2=\bm{p}_3=(1/2,1/2,0)$.
Then, the total reported demands of the three objects are $2$, $1$, and $0$, respectively.
Hence, $M_>=\{1\}$, $M_= = \{2\}$, and $M_<=\{3\}$.
By \Cref{lem:uo-criterion}, the set of utilitarian-optimal allocations is
\begin{align}
\mathcal{U}(P)=
\left\{
\begin{pmatrix}
1-a-b & 0   & a+b\\
a     & 1/2 & 1/2-a\\
b     & 1/2 & 1/2-b
\end{pmatrix}
\,\middle|\, 
\begin{array}{l}
0\le a\le 1/2,\\    
0\le b\le 1/2\\
\end{array}
\right\}.\label{eq:ex_utilopt}
\end{align}
\end{example}

\subsection{The WF and QP Mechanisms}
\label{subsec:wf-qp-mechanisms}

We study two natural mechanisms for selecting a desirable allocation from $\mathcal U(P)$.
The first mechanism is based on a water-filling procedure, while the second uses a quadratic tie-breaking criterion over the utilitarian-optimal set of allocations.

\begin{definition}
\label{def:wf-mechanism}
The \emph{water filling (WF) mechanism}, denoted by $\mechWF$, is defined as
follows.
Fix a profile $P\in\Delta^N$. Initialize $x_{ij}\coloneqq 0$ for every
$(i,j)\in N\times M$.

\textit{Phase~1 (fill up to caps).}
Increase all entries $(i,j)\in N\times M$ simultaneously at a rate of one unit per unit time, and stop increasing $x_{ij}$ as soon
as at least one of the following two events occurs: $x_{ij}=p_{ij}$ or $\sum_{r\in N}x_{rj}=1$.
Continue until no further increase is possible without violating one of these conditions, and let
$Y=(y_{ij})_{i\in N,\,j\in M}$ denote the resulting matrix.

\textit{Phase~2 (complete the matrix).}
Starting from $X\coloneqq Y$, increase all entries $(i,j)\in N\times M$ such that 
$\sum_{k\in M}x_{ik}<1$ and $\sum_{r\in N}x_{rj}<1$
at the same speed, and stop increasing $x_{ij}$ as soon as
either $\sum_{k\in M}x_{ik}=1$ or $\sum_{r\in N}x_{rj}=1$ holds.
Continue until $X$ is doubly stochastic, and set $\mechWF(P)\coloneqq X$.
\end{definition}

To understand the motivation behind the WF mechanism, note that the overlap representation of $\ell_1$~preferences suggests interpreting $p_{ij}$ as a per-object cap: agent $i$ derives
one unit of utility from each unit of object $j$ up to $p_{ij}$, and derives zero marginal utility
beyond $p_{ij}$.
Observe also that the matrix after Phase~1 has a column-wise closed form. 
In particular, for $j\in M_\le$, we have $y_{ij}=p_{ij}$ for every $i\in N$. 
For $j\in M_>$, let $t_j$ be the unique value satisfying $\sum_{i\in N}\min\{p_{ij},t_j\}=1$; then, $y_{ij}=\min\{p_{ij},t_j\}$ for every $i\in N$.

\begin{example}\label{ex:WF}
Consider applying the WF mechanism to the profile $P$ defined in \Cref{ex:small_instance}.
At the end of Phase~1, it outputs $Y$ such that $\bm{y}_1=(1/3,0,0)$ and $\bm{y}_2=\bm{y}_3=(1/3,1/2,0)$.
In Phase~2, it fills object $3$ and outputs
\begin{align}
\mechWF(P)=\begin{pmatrix}
1/3&0  &2/3\\
1/3&1/2&1/6\\
1/3&1/2&1/6
\end{pmatrix}.
\end{align}
\end{example}

\begin{definition}
\label{def:qp-mechanism}
The \emph{quadratic programming (QP) mechanism}, denoted by $\mechQP$, selects the
minimum $\ell_2$-norm allocation within the set of utilitarian-optimal allocations:
\begin{align}
\mechQP(P)\coloneqq
\operatorname*{arg\,min}_{X\in\mathcal{U}(P)}
\frac12\sum_{i\in N}\sum_{j\in M}x_{ij}^2. \label{eq:QP}
\end{align}
Note that the factor $1/2$ does not affect the mechanism and is included only for convenience in later calculations.
Since $\mathcal{U}(P)$ is nonempty, closed, and convex, and the objective is strictly convex, the minimum $\ell_2$-norm point is unique and the mechanism is well-defined.
\end{definition}

\begin{example}\label{ex:QP}
Consider applying the QP mechanism to the profile $P$ defined in \Cref{ex:small_instance}.
The symmetry between agents $2$ and $3$ implies $a=b$ in \eqref{eq:ex_utilopt}.
Writing $a=b=z$, the $\ell_2$-norm objective is $(1-2z)^2/2+z^2/2+z^2/2+1/8+1/8+(2z)^2/2+(1/2-z)^2/2+(1/2-z)^2/2=6(z-1/4)^2+5/8$. 
Hence, it is minimized at $z=1/4$, and since this yields a feasible allocation, we obtain
\begin{align}
\mechQP(P)=\begin{pmatrix}
1/2&0  &1/2\\
1/4&1/2&1/4\\
1/4&1/2&1/4
\end{pmatrix}.
\end{align}
\end{example}

Both mechanisms can be implemented in polynomial time. 
For the WF mechanism, Phase~1 can be computed column by column by sorting the $n$ caps in each column.
In Phase~2, we can determine the entry $(i,j)$ that the process would stop first; since each step fills at least one row or column to its capacity of~$1$, this phase requires at most $2n$ steps. 
For the QP mechanism, since \eqref{eq:QP} is a strictly convex separable quadratic minimization problem over the polytope \eqref{eq:mathcalU}, it can be solved by standard polynomial-time convex optimization methods. 

We observe that, since the descriptions of WF and QP do not distinguish between agents with the same preference and the output of each mechanism is unique, both mechanisms satisfy equal treatment of equals.

\begin{proposition}
\label{prop:equal-treatment}
The mechanisms $\mechWF$ and $\mechQP$ satisfy equal treatment of equals.
\end{proposition}

\section{Properties of the WF Mechanism}
\label{sec:wf-mechanism}

We begin by analyzing properties of the WF mechanism $\mechWF$ in this section.
First, the following proposition formalizes the fact that Phase~2 never changes the overlap utilities of the agents.
During the execution of the mechanism, we say that a row~$i$ (resp., a column~$j$) is ``full'' if $\sum_{k\in M}x_{ik} = 1$ (resp., $\sum_{r\in N}x_{rj} = 1$).

\begin{proposition}
\label{prop:wf-min-equals-phase1}
Let $P\in\Delta^N$, and let $Y$ and $\mechWF(P)$ be as in \Cref{def:wf-mechanism}.
Then, it holds that $\min\{\mechWF(P)_{ij},p_{ij}\}=y_{ij}$ for every $(i,j)\in N\times M$.
In particular,  $u_i(\mechWF(P)_i,\bm{p}_i)=\sum_{j\in M}y_{ij}$ for every $i\in N$.
\end{proposition}

\begin{proof}
Fix $(i,j)\in N\times M$.
Since the amounts filled in Phase~1 never exceed the caps, we have $y_{ij} \le p_{ij}$.
If $\mechWF(P)_{ij}=y_{ij}$, then $\mechWF(P)_{ij}\le p_{ij}$, and so
$\min\{\mechWF(P)_{ij},p_{ij}\}=y_{ij}$.
Otherwise, $\mechWF(P)_{ij}>y_{ij}$, which means that the $(i,j)$ entry is increased in Phase~2.
By the definition of Phase~2, at the start of this phase we have 
$\sum_{k\in M}y_{ik}<1$ and
$\sum_{r\in N}y_{rj}<1$.
Since column $j$ is not full at the start of Phase~2, the full-column event of column $j$ has not occurred in Phase~1.
Hence, the entry $(i,j)$ must have been stopped by the cap, that is, $y_{ij}=p_{ij}$.
Therefore, $\min\{\mechWF(P)_{ij},p_{ij}\}=p_{ij}=y_{ij}$.

Fixing $i\in N$ and summing over $j\in M$ gives the second claim.
\end{proof}

We now establish the main properties of the WF mechanism.
\begin{theorem}\label{thm:wf-props}
The mechanism $\mechWF$ is utilitarian-optimal, envy-free, and group-strategyproof.
\end{theorem}

The remainder of this section is devoted to proving each property in turn: utilitarian-optimality in \Cref{prop:wf-uo}, envy-freeness in \Cref{prop:wf-ef}, and group-strategyproofness in \Cref{prop:wf-gsp}.

\begin{proposition}\label{prop:wf-uo}
The mechanism $\mechWF$ is utilitarian-optimal.
\end{proposition}

\begin{proof}
Fix $P\in\Delta^N$, and let $X\coloneqq \mechWF(P)$.
We show that $X$ satisfies the characterization in \Cref{lem:uo-criterion}.
For $j\in M_\ge$, column $j$ becomes full in Phase~1, and Phase~2 never changes full columns.
Thus, $x_{ij}\le p_{ij}$ for every $i\in N$.
For $j\in M_\le$, column $j$ cannot become full before every cap $p_{ij}$ is reached, which means that $y_{ij}=p_{ij}$ for every $i\in N$, and Phase~2 only increases entries.
Thus, $x_{ij}\ge p_{ij}$ for every $i\in N$.
Therefore, \Cref{lem:uo-criterion} implies that $X$ is utilitarian-optimal.
\end{proof}

\begin{proposition}\label{prop:wf-ef}
The mechanism $\mechWF$ is envy-free.
\end{proposition}
\begin{proof}
Fix $P\in\Delta^N$, let $X\coloneqq \mechWF(P)$, and fix any agents $i,i'\in N$.
It suffices to show that
$u_i(\bm{x}_i,\bm{p}_i)\ge u_i(\bm{x}_{i'},\bm{p}_i)$.
We will show the object-wise domination
\begin{align}
\min\{x_{ij},p_{ij}\}\ge \min\{x_{i'j},p_{ij}\}\qquad (\forall j\in M). \label{eq:object-domination}
\end{align}

Fix $j\in M$.
If $x_{ij}\ge p_{ij}$, then $\min\{x_{ij},p_{ij}\}=p_{ij}\ge \min\{x_{i'j},p_{ij}\}$.
Suppose now that $x_{ij}<p_{ij}$.
Let $Y$ be the matrix at the end of Phase~1.
By \Cref{prop:wf-min-equals-phase1}, we have $y_{ij}=x_{ij}<p_{ij}$.
This implies that the increase of entry $(i,j)$ in Phase~1 stops because column $j$ becomes full.
Thus, $y_{rj}\le y_{ij}$ for every $r\in N$.
Moreover, Phase~2 never changes column $j$, so $x_{rj} = y_{rj}$ for every $r\in N$.
Hence, we have $\min\{x_{i'j},p_{ij}\} \le x_{i'j}=y_{i'j}\le y_{ij}=x_{ij}=\min\{x_{ij},p_{ij}\}$.
Therefore, \eqref{eq:object-domination} always holds.
Summing~\eqref{eq:object-domination} over $j\in M$ gives $u_i(\bm{x}_i,\bm{p}_i)\ge u_i(\bm{x}_{i'},\bm{p}_i)$, as desired.
\end{proof}

\begin{proposition}\label{prop:wf-gsp}
The mechanism $\mechWF$ is group-strategyproof.
\end{proposition}
\begin{proof}
Fix a profile $P\in\Delta^N$, a nonempty coalition $S\subseteq N$, and a joint misreport $\bm{q}_S\in\Delta^S$.
Write $Q\coloneqq(\bm{q}_S,\bm{p}_{-S})$.
Let $X\coloneqq\mechWF(P)$ and $X'\coloneqq\mechWF(Q)$.
For each object $j\in M$, define the coalition's overlap on $j$ by
$U_j(X)\coloneqq \sum_{i\in S}\min\{x_{ij},p_{ij}\}$ and 
$U_j(X')\coloneqq \sum_{i\in S}\min\{x'_{ij},p_{ij}\}$.
We will show that
$U_j(X')\le U_j(X)$ for all $j\in M$.

For $j\in M_\le$, we have $x_{ij}\ge p_{ij}$ for every $i\in N$ by the same reasoning as in \Cref{prop:wf-uo}.
Hence, $U_j(X)=\sum_{i\in S}p_{ij}$.
Since $\min\{x'_{ij},p_{ij}\}\le p_{ij}$, we obtain $U_j(X')\le U_j(X)$.

Now, let $j\in M_>$.
We analyze the outcome on column $j$ under the reported profile $Q$.

\medskip

\underline{Case 1}: Column $j$ is not full after Phase~1 under $Q$.
Let $Y'$ be the matrix obtained after Phase~1 under $Q$.
In this case, the full-column event never occurs on this column in Phase~1.
This implies that every entry on column $j$ reaches its cap by the construction.
Hence, $y'_{ij}=q_{ij}$ for every $i\in N$.
Moreover, since Phase~2 only increases entries, $x'_{ij}\ge y'_{ij}=q_{ij}=p_{ij}$ holds for every outside agent $i\in N\setminus S$.
Since $\sum_{i\in N}x'_{ij}=1$, we obtain
$U_j(X')
  \le \sum_{i\in S}x'_{ij}
  = 1-\sum_{i\in N\setminus S}x'_{ij}
  \le 1-\sum_{i\in N\setminus S}p_{ij}$.
On the other hand, since $j\in M_>$, column $j$ is full after Phase~1 under the truthful profile
$P$, and therefore $x_{ij}\le p_{ij}$ for every $i\in N$.
Thus,
$U_j(X)
  =\sum_{i\in S}x_{ij}
  =1-\sum_{i\in N\setminus S}x_{ij}
  \ge 1-\sum_{i\in N\setminus S}p_{ij}$.
Combining the two inequalities yields $U_j(X')\le U_j(X)$.

\medskip

\underline{Case 2}: Column $j$ is full after Phase~1 under $Q$.
In this case, Phase~2 never changes column~$j$.
Let $t$ and $t'$ be the times at which column~$j$ becomes full in Phase~1 under $P$ and $Q$, respectively.
Since Phase~1 is column-wise capped water filling, the equal-speed filling implies that $x_{ij}=\min\{p_{ij},t\}$ for all $i\in N$, $\sum_{i\in N}\min\{p_{ij},t\}=1$, and $U_j(X)=\sum_{i\in S}\min\{p_{ij},t\} = \sum_{i\in S}x_{ij}$.
Similarly, we have $x'_{ij}=\min\{q_{ij},t'\}$ for all $i\in N$, and $\sum_{i\in N}\min\{q_{ij},t'\}=1$.
Hence, $U_j(X')=\sum_{i\in S}\min\{p_{ij},q_{ij},t'\}$.

If $t\ge t'$, then
\begin{align}
U_j(X')
=\sum_{i\in S}\min\{p_{ij},q_{ij},t'\}
\le\sum_{i\in S}\min\{p_{ij},t'\}
\le\sum_{i\in S}\min\{p_{ij},t\}
=U_j(X).
\end{align}
Otherwise (i.e., $t<t'$), we have
\begin{align}
U_j(X')
&\le \sum_{i\in S}x'_{ij}
= 1-\sum_{i\in N\setminus S}x'_{ij}
= 1-\sum_{i\in N\setminus S}\min\{p_{ij},t'\}\\
&\le 1-\sum_{i\in N\setminus S}\min\{p_{ij},t\}
=1-\sum_{i\in N\setminus S}x_{ij}
=\sum_{i\in S}x_{ij}
=U_j(X).
\end{align}
Therefore, $U_j(X')\le U_j(X)$ holds regardless of the relation between $t$ and $t'$.

\medskip

We have shown that $U_j(X')\le U_j(X)$ for every $j\in M$. 
Since $\sum_{i\in S}u_i(\bm{x}_i,\bm{p}_i)=\sum_{j\in M}U_j(X)$ and $\sum_{i\in S}u_i(\bm{x}'_i,\bm{p}_i)=\sum_{j\in M}U_j(X')$, we have
$\sum_{i\in S}u_i(\bm{x}'_i,\bm{p}_i)\le \sum_{i\in S}u_i(\bm{x}_i,\bm{p}_i)$.
Hence, it cannot happen that every member of $S$ weakly improves in $X'$ compared to $X$ and at least one member strictly improves.
This implies that $\mechWF$ is group-strategyproof.
\end{proof}

\section{Properties of the QP Mechanism}
\label{sec:qp-mechanism}

In this section, we turn our attention to the QP mechanism and establish the following properties.
\begin{theorem}\label{thm:qp-props}
The mechanism $\mechQP$ is utilitarian-optimal, envy-free, and strategyproof.
\end{theorem}
Utilitarian-optimality follows directly from the definition of the mechanism. 
We will show envy-freeness in \Cref{prop:qp-ef} and strategyproofness in \Cref{prop:qp-sp}.
Note that unlike the WF mechanism, $\mechQP$ does not even satisfy weak group-strategyproofness (see Appendix~\ref{sec:weak-gsp}).

\subsection{Optimality Conditions}
\label{subsec:qp-optimality-conditions}

By \Cref{lem:uo-criterion}, an equivalent definition of $\mechQP$ is that, for any $P\in\Delta^N$, the output $\mechQP(P)$ is the unique optimal solution to the following convex quadratic
optimization problem:
 \begin{alignat}{2}
({\rm QP})\quad
\min \quad & \textstyle\frac{1}{2}\sum_{i\in N}\sum_{j\in M}x^2_{ij} \\
\text{s.t.}\quad & \textstyle\sum_{j\in M}x_{ij}=1  &\quad& (i\in N),\\
                 & \textstyle\sum_{i\in N}x_{ij}=1  && (j\in M),\\
                 & x_{ij}\le p_{ij}                 && (i\in N,\ j\in M_\ge),\\
                 & x_{ij}\ge p_{ij}                && (i\in N,\ j\in M_\le).
\end{alignat}
As we show in \Cref{cor:nonneg} later, the optimizer always has coordinates in $[0,1]$, so
nonnegativity constraints are redundant.
Consider the following Lagrangian function for {\rm(QP)}:
\begin{align*}
\frac{1}{2}\sum_{i\in N}\sum_{j\in M}x_{ij}^2
&+ \sum_{i\in N}\alpha_i\!\left(1-\sum_{j\in M}x_{ij}\right)
 + \sum_{j\in M}\beta_j\!\left(1-\sum_{i\in N}x_{ij}\right) \notag\\
&+ \sum_{i\in N}\sum_{j\in M_{\ge}}\lambda_{ij}(x_{ij}-p_{ij})
 + \sum_{i\in N}\sum_{j\in M_{\le}}\mu_{ij}(p_{ij}-x_{ij}).
\end{align*}
The primal variable is $X=(x_{ij})_{i\in N,\,j\in M}$.
The multipliers are $(\alpha_i)_{i\in N}$, $(\beta_j)_{j\in M}$,
$(\lambda_{ij})_{(i,j)\in N\times M}$, and $(\mu_{ij})_{(i,j)\in N\times M}$.
The latter multipliers are defined on all of $N\times M$, and we later set
$\lambda_{ij}=0$ for $j\in M_<$ and $\mu_{ij}=0$ for $j\in M_>$ via condition (h).
The \emph{Karush--Kuhn--Tucker (KKT)} conditions for the problem are:
\begin{itemize}
\setlength{\leftskip}{3mm}
    \item[(a)~] $x_{ij}-\alpha_i-\beta_j+\lambda_{ij}-\mu_{ij}=0$
    for $i\in N$ and $j\in M$,
    \item[(b)~] $\sum_{j\in M}x_{ij}=1$ for $i\in N$,
    \item[(c)~] $\sum_{i\in N}x_{ij}=1$ for $j\in M$,
    \item[(d)~] $x_{ij}\le p_{ij}$ for $i\in N$ and $j\in M_\ge$,
    \item[(e)~] $x_{ij}\ge p_{ij}$ for $i\in N$ and $j\in M_\le$,
\item[(f)~] $\lambda_{ij}\ge 0$ and $\lambda_{ij}(x_{ij}-p_{ij})=0$ for $i\in N$ and $j\in M_\ge$,
\item[(g)~] $\mu_{ij}\ge 0$ and $\mu_{ij}(p_{ij}-x_{ij})=0$ for $i\in N$ and $j\in M_\le$,
\item[(h)~] $\lambda_{ij}=0$ for $j\notin M_\ge$ and $\mu_{ij}=0$ for $j\notin M_\le$.
\end{itemize}
In particular, for $j\in M_=$, both inequality families are present, and therefore both
multipliers $(\lambda_{ij},\mu_{ij})$ enter the same stationarity equation (a).
Accordingly, stationarity is written in the unified form
$x_{ij}-\alpha_i-\beta_j+\lambda_{ij}-\mu_{ij}=0$ over all $(i,j) \in N\times M$.

The KKT conditions give the following coordinatewise representation of the QP allocation, which is
the main tool for the rest of this section.
\begin{proposition}\label{prop:opt-rep}
There exist constants $(\alpha^*_i)_{i\in N}$ and $(\beta^*_j)_{j\in M}$ such that the unique optimal solution $(x^*_{ij})_{i\in N,\,j\in M}$ to problem (QP) can be represented as
\begin{align}
    x^*_{ij}&=\begin{cases}
        p_{ij}           & \text{if }j\in M_=,\\
        \min\{\alpha^*_i+\beta^*_j,\,p_{ij}\} & \text{if }j\in M_>,\\
        \max\{\alpha^*_i+\beta^*_j,\,p_{ij}\} & \text{if }j\in M_<.
    \end{cases}
    \label{eq:opt-sol}
\end{align}
\end{proposition}
\begin{proof}
By standard convex-programming KKT theory, the optimal solution
$(x^*_{ij})_{i\in N,\,j\in M}$ satisfies the KKT conditions (a)--(h); see, for example, the book by \citet{BoydVa04}.
Let $(\alpha^*_i)_{i\in N}$ and $(\beta^*_j)_{j\in M}$ be the corresponding multipliers from the KKT conditions.
Fix any $i\in N$ and $j\in M$.
If $j\in M_=$, conditions (d) and (e) immediately imply $x^*_{ij}=p_{ij}$.

In case $j\in M_>$, we consider two subcases.
If $\lambda_{ij}=0$, then $x^*_{ij}=\alpha^*_i+\beta^*_j$ by (a), and
$p_{ij}\ge x^*_{ij}=\alpha^*_i+\beta^*_j$ by (d).
If $\lambda_{ij}>0$, then $x^*_{ij}=p_{ij}$ by (f), and
$\alpha^*_i+\beta^*_j>x^*_{ij}=p_{ij}$ follows from~(a).
Thus, in each subcase, $x^*_{ij} = \min\{\alpha^*_i+\beta^*_j,\,p_{ij}\}$.

Similarly, in case $j\in M_<$, we consider two subcases.
If $\mu_{ij}=0$, then $x^*_{ij}=\alpha^*_i+\beta^*_j$ by (a), and
$p_{ij}\le x^*_{ij}=\alpha^*_i+\beta^*_j$ by (e).
If $\mu_{ij}>0$, then $x^*_{ij}=p_{ij}$ by (g), and
$\alpha^*_i+\beta^*_j<x^*_{ij}=p_{ij}$ follows from~(a).
Thus, in each subcase, $x^*_{ij} = \max\{\alpha^*_i+\beta^*_j,\,p_{ij}\}$.
\end{proof}

The next lemma 
shows that if an agent is under-allocated relative to her report on some $j\in M_>$ and over-allocated on
some $j'\in M_<$, then her allocation on $j$ is at least her allocation on $j'$.
\medskip

\begin{lemma}\label{lem:relation}
    Let $(x^*_{ij})_{i\in N,\,j\in M}$ be the optimal solution to problem (QP).
    For each $j\in M_>$, $j' \in M_<$, and $i\in N$, if
    $p_{ij} > x^*_{ij}$ and $x^*_{ij'} > p_{ij'}$, then
    $x^*_{ij} \ge x^*_{ij'}$.
\end{lemma}
\begin{proof}
    Let $(\alpha^*_i)_{i\in N}$ and $(\beta^*_j)_{j\in M}$ be the constants from \Cref{prop:opt-rep}.
    Take an arbitrary agent $i' \neq i$.
    First, since $x^*_{ij} < p_{ij}$ and $j\in M_>$, we have $x^*_{ij}=\min\{\alpha^*_i+\beta^*_j, p_{ij}\}=\alpha^*_i + \beta^*_j$.
    Moreover, $x^*_{i'j} \leq \alpha^*_{i'}+\beta^*_j$.
    Thus, $x^*_{i'j} - x^*_{ij} \leq (\alpha^*_{i'}+\beta^*_j) - (\alpha^*_{i}+\beta^*_j) = \alpha^*_{i'}-\alpha^*_i$.
    Second, since $x^*_{ij'} > p_{ij'}$ and $j' \in M_<$, we have $x^*_{ij'}=\max\{\alpha^*_i+\beta^*_{j'},p_{ij'}\} = \alpha^*_i+\beta^*_{j'}$.
    Moreover, $x^*_{i'j'} \geq \alpha^*_{i'}+\beta^*_{j'}$.
    Thus, $x^*_{i'j'} - x^*_{ij'} \geq (\alpha^*_{i'}+\beta^*_{j'}) - (\alpha^*_{i}+\beta^*_{j'}) = \alpha^*_{i'}-\alpha^*_i$.

    By combining these observations, we obtain $x^*_{i'j'} - x^*_{ij'} \geq x^*_{i'j} - x^*_{ij}$ for every $i'\in N$, or equivalently, $x^*_{ij} - x^*_{i'j} \geq x^*_{ij'} - x^*_{i'j'}$.
    Summing this inequality over all $i' \neq i$, we obtain
    \begin{align*}
    \MoveEqLeft
        (n-1)x^*_{ij} - \sum_{i'\neq i} x^*_{i'j} \geq (n-1)x^*_{ij'} - \sum_{i'\neq i} x^*_{i'j'}.
    \end{align*}
    Since $\sum_{i'\in N} x^*_{i'j} = \sum_{i'\in N} x^*_{i'j'}=1$, this implies that $(n-1)x^*_{ij} - (1-x^*_{ij}) \geq  (n-1)x^*_{ij'} - (1 - x^*_{ij'})$.
    It follows that $x^*_{ij} \geq x^*_{ij'}$, as desired.
\end{proof}

\Cref{lem:relation} implies that the QP allocation for each agent stays within the
range of that agent's reported amounts.
\begin{lemma}\label{lem:maxmin}
    Let $(x^*_{ij})_{i\in N,\,j\in M}$ be the optimal solution to problem (QP).
    For each agent $i\in N$ and object $j\in M$, we have $\min_{k\in M}p_{ik}\le x^*_{ij}\le \max_{k\in M}p_{ik}$.
\end{lemma}
\begin{proof}
  Assume, for the sake of contradiction, that $x^*_{ij}>\max_{k\in M}p_{ik}$ for some $i\in N$ and $j\in M$.
  In particular, $x^*_{ij}>p_{ij}$.
  By conditions (d) and (e), we have $j\in M_<$.
  Since $\sum_{k\in M} x^*_{ik} = \sum_{k\in M} p_{ik} = 1$, there exists an object~$j'$ such that $x^*_{ij'} < p_{ij'}$.
  Thus, by conditions (d) and (e), we have $j'\in M_>$.
  By \Cref{lem:relation}, we obtain $x^*_{ij}\le x^*_{ij'}<p_{ij'}\le \max_{k\in M}p_{ik}$.
  This is a contradiction.

  Similarly, suppose for the sake of contradiction that $x^*_{ij}<\min_{k\in M}p_{ik}$ for some $i\in N$ and $j\in M$.
  In particular, $x^*_{ij}<p_{ij}$.
  By conditions (d) and (e), we have $j\in M_>$.
  Since $\sum_{k\in M} x^*_{ik} = \sum_{k\in M} p_{ik} = 1$, there exists an object $j'$ such that $x^*_{ij'} > p_{ij'}$.
  Thus, by conditions (d) and (e), we have $j'\in M_<$.
  By \Cref{lem:relation}, we obtain $x^*_{ij}\ge x^*_{ij'}>p_{ij'}\ge \min_{k\in M}p_{ik}$.
  This is a contradiction.
\end{proof}

The corollary below justifies the formulation of ${\rm(QP)}$ without explicit nonnegativity
constraints.
\begin{corollary}\label{cor:nonneg}
The optimal solution $(x^*_{ij})_{i\in N,\,j\in M}$ to problem ${\rm(QP)}$ is a fractional allocation.
\end{corollary}
\begin{proof}
By feasibility of problem ${\rm(QP)}$, 
we have $\sum_{j\in M}x^*_{ij}=1$ for all $i\in N$ and $\sum_{i\in N}x^*_{ij}=1$ for all $j\in M$. 
Additionally, by \Cref{lem:maxmin}, we have $x^*_{ij}\ge \min_{k\in M}p_{ik}\ge 0$ for every $i\in N$ and $j\in M$.
Thus, $(x^*_{ij})_{i\in N,\,j\in M}$ is a fractional allocation.
\end{proof}

\subsection{Envy-Freeness}
\label{subsec:qp-envy-freeness}

We next use the coordinatewise KKT representation to show that the QP allocation is envy-free.
\begin{proposition}\label{prop:qp-ef}
The mechanism $\mechQP$ is envy-free.
\end{proposition}
\begin{proof}
Let $P\in\Delta^N$ and $X^*\coloneqq\mechQP(P)$, and let $(\alpha^*_i)_{i\in N}$ and $(\beta^*_j)_{j\in M}$ be the constants given in \Cref{prop:opt-rep}.
To prove envy-freeness, fix two agents $i,i'\in N$.
We will show that $\du(\bm{x}^*_{i'},\bm{p}_i)-\du(\bm{x}^*_i,\bm{p}_i)\ge 0$.
We divide the cases according to whether $\alpha_i^*\ge \alpha^*_{i'}$ or $\alpha_i^*< \alpha^*_{i'}$.
For each object $j\in M$, we evaluate $|x^*_{i'j}-p_{ij}|-|x^*_{ij}-p_{ij}|$, distinguishing the following cases:
(a) $x^*_{ij}= p_{ij}$,
(b)~$x^*_{ij}< p_{ij}$, or
(c) $x^*_{ij}> p_{ij}$.

\medskip

\underline{Case 1}: $\alpha^*_i \geq \alpha^*_{i'}$.
For object $j\in M$ satisfying (a), we have
\begin{align}
    |x^*_{i'j}-p_{ij}|-|x^*_{ij}-p_{ij}|=|x^*_{i'j}-x^*_{ij}|\ge x^*_{i'j}-x^*_{ij}.
    \label{eq:EF>=}
\end{align}
For object $j\in M$ satisfying (b), 
\Cref{prop:opt-rep} implies that $j\in M_>$ and 
$x^*_{i'j}=\min\{\alpha^*_{i'}+\beta^*_j,p_{i'j}\}\le \alpha^*_{i'}+\beta^*_j\le\alpha^*_i+\beta^*_j = x^*_{ij} < p_{ij}$.
Thus, we get
\begin{align}
|x^*_{i'j}-p_{ij}|-|x^*_{ij}-p_{ij}|
= (p_{ij}-x^*_{i'j})-(p_{ij}-x^*_{ij})
= x^*_{ij}-x^*_{i'j}
\ge 0 \ge x^*_{i'j}-x^*_{ij}.
\label{eq:EF>>}
\end{align}
For object $j\in M$ satisfying (c), we have
\begin{align}
    |x^*_{i'j}-p_{ij}|-|x^*_{ij}-p_{ij}|
    \ge (x^*_{i'j}-p_{ij})-(x^*_{ij}-p_{ij})
    = x^*_{i'j}-x^*_{ij}.
    \label{eq:EF><}
\end{align}
Therefore, by combining~\eqref{eq:EF>=}, \eqref{eq:EF>>}, and \eqref{eq:EF><}, we obtain
\begin{align*}
\du(\bm{x}^*_{i'},\bm{p}_i)-\du(\bm{x}^*_i,\bm{p}_i)
&=\sum_{j\in M}(|x^*_{i'j}-p_{ij}|-|x^*_{ij}-p_{ij}|)\\
&\ge\sum_{j\in M}(x^*_{i'j}-x^*_{ij})
=\sum_{j\in M}x^*_{i'j}-\sum_{j\in M}x^*_{ij}
=1-1
=0.
\end{align*}

\medskip

\underline{Case 2}: $\alpha^*_i < \alpha^*_{i'}$.
For object $j\in M$ satisfying (a), we have
\begin{align}
    |x^*_{i'j}-p_{ij}|-|x^*_{ij}-p_{ij}|=|x^*_{i'j}-x^*_{ij}|\ge x^*_{ij}-x^*_{i'j}.
    \label{eq:EF<=}
\end{align}
For object $j\in M$ satisfying (b), we have
\begin{align}
    |x^*_{i'j}-p_{ij}|-|x^*_{ij}-p_{ij}|
    \ge (p_{ij}-x^*_{i'j})-(p_{ij}-x^*_{ij})
    = x^*_{ij}-x^*_{i'j}.
    \label{eq:EF<>}
\end{align}
For object $j\in M$ satisfying (c), 
\Cref{prop:opt-rep} implies that $j\in M_<$ and $x^*_{ij}=\max\{\alpha^*_{i}+\beta^*_j,p_{ij}\}=\alpha^*_{i}+\beta^*_j> p_{ij}$.
Additionally, we have $x^*_{i'j}=\max\{\alpha^*_{i'}+\beta^*_j,p_{i'j}\}\ge \alpha^*_{i'}+\beta^*_j$, and since $\alpha^*_i < \alpha^*_{i'}$, we obtain $x^*_{i'j} \geq \alpha^*_{i'}+\beta^*_j > \alpha^*_i+\beta^*_j = x^*_{ij} > p_{ij}$.
Thus,
\begin{align}
    |x^*_{i'j}-p_{ij}|-|x^*_{ij}-p_{ij}|
    =(x^*_{i'j}-p_{ij})-(x^*_{ij}-p_{ij})
    =x^*_{i'j}-x^*_{ij}
    \ge 0\ge x^*_{ij}-x^*_{i'j}.
    \label{eq:EF<<}
\end{align}
Therefore, combining~\eqref{eq:EF<=}, \eqref{eq:EF<>}, and \eqref{eq:EF<<}, we obtain
\begin{align*}
\du(\bm{x}^*_{i'},\bm{p}_i)-\du(\bm{x}^*_i,\bm{p}_i)
&=\sum_{j\in M}(|x^*_{i'j}-p_{ij}|-|x^*_{ij}-p_{ij}|)\\
&\ge\sum_{j\in M}(x^*_{ij}-x^*_{i'j})
=\sum_{j\in M}x^*_{ij}-\sum_{j\in M}x^*_{i'j}
=1-1
=0.
\end{align*}

\medskip

The two cases together complete the proof.
\end{proof}

\subsection{Strategyproofness}
\label{sec:proof-strategyproofness}

The main incentive result for the QP mechanism is strategyproofness.
\begin{proposition}
\label{prop:qp-sp}
The mechanism $\mechQP$ is strategyproof.
\end{proposition}

The proof follows a one-dimensional path argument and is given in Appendix~\ref{app:prop-qp-sp}.
At a high level, we connect the truthful report $\bm{p}_i$
to an arbitrary report $\bm{q}_i$ by a line segment $\bm{q}_i(\tau)=\bm{p}_i+\tau(\bm{q}_i-\bm{p}_i)$ for $\tau\in[0,1]$, and track the QP outcome $X(\tau)$ along this segment.
The key observation is that this path can be decomposed into finitely many intervals on which the relevant combinatorial structure of the QP remains fixed. 
On each such interval, both $X(\tau)$ and $\du(\bm{x}_i(\tau),\bm{p}_i)$ are affine.
Thus, it suffices to show that the slope of $\du(\bm{x}_i(\tau),\bm{p}_i)$ is nonnegative on every interval.

Fix one such interval, and let $Y=\dd{X(\tau)}/\dd{\tau}$ denote the constant velocity matrix on it. 
By differentiating the KKT conditions of the QP while keeping the active constraints fixed,
one can characterize $Y$ as the minimum $\ell_2$-norm point satisfying the tangent constraints.
Since the row and column sums of $Y$ are zero, this velocity matrix can be viewed as a circulation in the bipartite agent--object graph. 
We show that, if the slope were negative, then there would exist a directed cycle such that subtracting a small multiple of this cycle from $Y$ preserves all tangent constraints while strictly decreasing the $\ell_2$-norm of $Y$, thereby contradicting the minimum-norm characterization of the QP velocity.
Hence, every interval has a nonnegative slope, and summing over the finitely many intervals yields the desired conclusion.

\section{Comparing the WF and QP Mechanisms}
\label{sec:wf-qp-comparison}

The WF and QP mechanisms are both utilitarian-optimal,
envy-free, and strategyproof, and satisfy equal treatment of equals. 
We have seen that the WF mechanism satisfies group-strategyproofness, while the QP mechanism fails even weak group-strategyproofness.
In this section, we highlight an advantage of the QP mechanism by showing that it provides better guarantees in terms of \emph{egalitarian overlap welfare}, the smallest overlap utility among all agents.

We first show that the WF mechanism can have much smaller egalitarian overlap welfare than the
QP mechanism.
For an allocation $X$, write
$\EG(X,P)\coloneqq \min_{i\in N}u_i(\bm{x}_i,\bm{p}_i)$ for its egalitarian overlap welfare with respect to profile~$P$. 

\begin{proposition}
\label{prop:wf-linearly-worse}
For each $n\ge 3$, there exists a profile $P\in\Delta^N$ such that
$\EG(\mechWF(P),P)=1/n$ and $\EG(\mechQP(P),P)=1/2$. Consequently,
$\EG(\mechQP(P),P)/\EG(\mechWF(P),P)=n/2$.
\end{proposition}

\begin{proof}
Let $N=M=\{1,\dots,n\}$. 
Let $\bm{p}_1=(1,0,\dots,0)$ and, for each $i \in \{2,\dots,n\}$, let
$\bm{p}_i=(1/(n-1),\dots,1/(n-1),0)$, where the first $n-1$ coordinates are $1/(n-1)$.
Then, object $1$ is over-demanded, objects $2,\dots,n-1$ are exactly demanded, and object $n$
is under-demanded.

By \Cref{lem:uo-criterion}, every symmetric utilitarian-optimal allocation has the following form for some
$a\in[0,1]$: $x_{11}=a$, $x_{1n}=1-a$, and, for every $i\in\{2,\dots,n\}$,
$x_{i1}=(1-a)/(n-1)$, $x_{ij}=1/(n-1)$ for $2\le j\le n-1$, and
$x_{in}=a/(n-1)$. 
The corresponding overlap utilities are $u_1(\bm{x}_1,\bm{p}_1)=a$ and
$u_i(\bm{x}_i,\bm{p}_i)=1-a/(n-1)$ for $i\in\{2,\dots,n\}$.

For the WF mechanism, object $1$ is filled at equal speed across all $n$ agents until the column is
full. 
Hence, each agent receives $1/n$ of object $1$ in Phase~1. Objects $2,\dots,n-1$ are filled
to the caps of agents $2,\dots,n$, and Phase~2 only completes object $n$. 
Therefore,
$u_1(\mechWF(P)_1,\bm{p}_1)=1/n$, whereas every other agent has an overlap utility of at least $1/n$, and hence $\EG(\mechWF(P),P)=1/n$.

For the QP mechanism, the objective on the utilitarian-optimal set is
$\frac12\left(1+\frac1{n-1}\right)\left(a^2+(1-a)^2\right)$
up to an additive constant
independent of $a$.
This is uniquely minimized at $a=1/2$. 
Hence,
$u_1(\mechQP(P)_1,\bm{p}_1)=1/2$ and
$u_i(\mechQP(P)_i,\bm{p}_i)=1-1/(2(n-1)) \ge 1/2$ for $i\in\{2,\dots,n\}$.
It follows that
$\EG(\mechQP(P),P)=1/2$.
\end{proof}

For the other direction, we prove that the egalitarian overlap welfare of the QP mechanism can be less than that of the WF mechanism by at most a factor of $2$. 

\begin{theorem}
\label{thm:wf-qp-two-factor}
For every profile $P\in\Delta^N$, it holds that $\EG(\mechWF(P),P)\le 2\cdot \EG(\mechQP(P),P)$.
\end{theorem}

\begin{proof}
Let $X^*\coloneqq\mechQP(P)$, let $z\coloneqq \EG(X^*,P)$, and choose
$i'\in N$ such that $u_{i'}(\bm{x}^*_{i'},\bm{p}_{i'})=z$.
If $z=1$, the claim is immediate. Assume therefore that $z<1$.

Let $(\alpha,\beta)$ be the KKT potentials from \Cref{prop:opt-rep}, with normalization\footnote{This normalization can be achieved by decreasing $\alpha_i$ by $\alpha_{i'}$ for all $i\in N$, and increasing $\beta_j$ by $\alpha_{i'}$ for all $j\in M$.}
$\alpha_{i'}=0$. Define $\widehat O\coloneqq \{j\in M_>\mid x^*_{i' j}<p_{i' j}\}$
and $\widehat U\coloneqq \{j\in M_<\mid x^*_{i' j}>p_{i' j}\}$.
Since $z<1$, agent $i'$ has a positive overlap deficit compared to the ideal overlap. 
By \Cref{lem:uo-criterion}, such a deficit can occur only on over-demanded objects,
and therefore $\widehat O\neq\emptyset$. 
Since the row sum of $\bm{x}^*_{i'}$ is $1$, this
deficit must be balanced by positive excess on under-demanded objects, and $\widehat U\neq\emptyset$ as well.

First, we claim that $\alpha_i\le z/|\widehat U|$ for every $i\in N$. Indeed, the row identity for
$i'$ gives
\begin{align*}
1=z+\sum_{k\in \widehat U}(x^*_{i' k}-p_{i' k})
=z+\sum_{k\in \widehat U}(\alpha_{i'}+\beta_k-p_{i' k})
=z+\sum_{k\in \widehat U}(\beta_k-p_{i' k}). 
\end{align*}
Hence,
$\sum_{k\in \widehat U}\beta_k=1-z+\sum_{k\in \widehat U}p_{i' k}\ge 1-z$.
For every $i\in N$, feasibility then gives
\begin{align*}
1
\ge \sum_{k\in \widehat U}x^*_{ik}
\ge \sum_{k\in \widehat U}(\alpha_i+\beta_k)
\ge |\widehat U|\alpha_i+1-z,
\end{align*}
which proves $\alpha_i\le z/|\widehat U|$.

For each $j\in M_>$, let $t_j$ be the water-filling threshold, that is,
$\sum_{i\in N}\min\{p_{ij},t_j\}=1$.
The bound on $\alpha_i$ implies that
$
x^*_{ij}
=\min\{p_{ij},\alpha_i+\beta_j\}
\le
\min\{p_{ij}, \, \beta_j+{z}/{|\widehat U|}\}$.
Summing over $i \in N$ and using $\sum_{i\in N} x^*_{ij}=1$, we obtain 
\begin{align*}
\sum_{i\in N} \min\{p_{ij}, t_j\} 
= 1
= \sum_{i\in N} x^*_{ij}
\le \sum_{i\in N}\min\{p_{ij}, \beta_j + z/|\widehat U|\}.
\end{align*}
Note that the function $f(\tau) \coloneqq \sum_{i\in N}\min\{p_{ij}, \tau\}$ is strictly increasing for $\tau \in [0, \max_{i\in N}p_{ij}]$.
Since $f(t_j)\le f(\beta_j+z/|\widehat U|)$, we have $t_j\le \beta_j+z/|\widehat U|$.

Next, we claim that $|\widehat O|/|\widehat U|\le z/(1-z)$.
Let $b\coloneqq\max_{k\in\widehat U}\beta_k$. 
By \Cref{lem:relation}, for every
$j\in\widehat O$ and $k\in\widehat U$, it holds that
$\beta_j=x^*_{i' j}\ge x^*_{i' k}=\beta_k$.
Therefore, $\beta_j\ge b$ for all $j\in\widehat O$, and
$z\ge \sum_{j\in\widehat O}x^*_{i' j}
=\sum_{j\in\widehat O}\beta_j\ge |\widehat O|b$. 
On the other hand,
$1-z=\sum_{k\in\widehat U}(x^*_{i' k}-p_{i' k})
=\sum_{k\in\widehat U}(\beta_k-p_{i' k})
\le \sum_{k\in\widehat U}\beta_k\le |\widehat U|b$.
Thus, we obtain $|\widehat O|/|\widehat U|\le z/(1-z)$.

Now, we compare the overlap utility of agent $i'$. 
For objects $j\in M_{\le}$, both mechanisms obtain
the full reported overlap $p_{i' j}$. 
Hence,
\begin{align*}
u_{i'}(\mechWF(P)_{i'},\bm{p}_{i'})-z
=
\sum_{j\in M_>}
\left(
\min\{p_{i' j},t_j\}
-
\min\{p_{i' j},\beta_j\}
\right).
\end{align*}
Terms with $j\in M_>\setminus\widehat O$ are nonpositive since $x^*_{i'j} = p_{i'j}$, which means that $\min\{p_{i'j},\beta_j\} = p_{i'j}$. 
For $j\in\widehat O$, we have
$x^*_{i'j} = \beta_j<p_{i' j}$, and hence
\begin{align*}
\min\{p_{i' j},t_j\}-\min\{p_{i' j},\beta_j\}
=\min\{p_{i' j},t_j\}-\beta_j
\le t_j - \beta_j
\le {z}/{|\widehat U|}.
\end{align*}
Therefore, we have
\begin{align*}
u_{i'}(\mechWF(P)_{i'},\bm{p}_{i'})-z
\le |\widehat O|\cdot {z}/{|\widehat U|}
\le z^2/(1-z).
\end{align*}
It follows that
$\EG(\mechWF(P),P)\le u_{i'}(\mechWF(P)_{i'},\bm{p}_{i'})\le z + z^2/(1-z) = z/(1-z)$.
Since overlap utilities are always at most~$1$, we have
\begin{align*}
\EG(\mechWF(P),P)\le \min\left\{1,\,z/(1-z)\right\}\le 2z
=2\cdot\EG(\mechQP(P),P),
\end{align*}
where the second inequality holds because $1\le2z$ if $z\ge1/2$, and 
$z/(1-z)\le2z$ if $z<1/2$.
\end{proof}

Despite the approximation provided by \Cref{thm:wf-qp-two-factor}, as the following example shows, the QP mechanism does not dominate the WF mechanism in terms of egalitarian overlap welfare.
\begin{example}
\label{ex:qp-wf-example}
Let $n=3$ and $N = M = \{1,2,3\}$, and consider the profile
\begin{align*}
P =
\begin{pmatrix}
0 & 2/5 & 3/5\\
2/5 & 2/5 & 1/5\\
2/5 & 2/5 & 1/5
\end{pmatrix}.
\end{align*}
Here, object $1$ is under-demanded, object $2$ is over-demanded, and object $3$ is exactly demanded.
Then, the WF mechanism and the QP mechanism respectively give
\begin{align*}
\mechWF(P)=
\begin{pmatrix}
1/15 & 1/3 & 3/5\\
7/15 & 1/3 & 1/5\\
7/15 & 1/3 & 1/5
\end{pmatrix}
\quad\text{and}\quad
\mechQP(P)=
\begin{pmatrix}
1/5 & 1/5 & 3/5\\
2/5 & 2/5 & 1/5\\
2/5 & 2/5 & 1/5
\end{pmatrix}.
\end{align*}
Thus, we have
${\EG(\mechWF(P),P)}/{\EG(\mechQP(P),P)}
={(14/15)}/{(4/5)}
={7}/{6}$.
\end{example}

Determining the smallest constant $c$ such that $\EG(\mechWF(P),P)\le c\cdot \EG(\mechQP(P),P)$ always holds remains an open question.
\Cref{thm:wf-qp-two-factor} and \Cref{ex:qp-wf-example} show that $7/6 \le c \le 2$.

\section{Conclusion and Future Work}

In this paper, we have investigated a fractional assignment setting with an equal number of agents and objects.
Unlike in the classic random assignment model, agents' ideal distributions in our setting need not be degenerate, but may instead be mixed over multiple objects.
The agents evaluate other distributions according to the $\ell_1$ distance---equivalently, the overlap utility---from their ideal distributions.
We propose two mechanisms that are utilitarian-optimal, envy-free, strategyproof, and satisfy equal treatment of equals, and demonstrate an advantage of each mechanism.

A natural direction for future work is to deepen our understanding of the class of utilitarian-optimal, envy-free, and strategyproof mechanisms satisfying equal treatment of equals.
In particular, it would be interesting to determine whether our proposed mechanisms are optimal in some sense within this class.
More broadly, our fractional assignment setting could be extended to other distance measures, such as $\ell_p$ distances for $p > 1$.
Since these measures do not admit a simple equivalent overlap representation, the structures of desirable mechanisms may differ significantly from those in our setting.

\section*{AI Declaration}

The authors used ChatGPT and Codex (OpenAI) to assist in developing the proof of \Cref{prop:qp-sp}. In particular, these tools helped us identify the tangent program approach used in the proof. We verified the final argument and take full responsibility for its correctness.

\section*{Acknowledgments}

This work was partially supported by 
JST ERATO Grant Number JPMJER2301, 
by JST CRONOS Grant Number JPMJCS24K2, 
by JSPS KAKENHI Grant Numbers JP25K00137,
JP21K17708, JP21H03397, JP26K14706, JP26K14718, and 
JP26K02867,
Japan, 
by the Ministry of Education, Singapore, under the Academic Research Fund Tier 1 (FY 2026) grant number 251RES2604, and by an NUS Start-up Grant.  
We thank the anonymous reviewers for their valuable feedback.

\bibliographystyle{plainnat}
\bibliography{main}

\appendix

\section{Comparison to the Work of \citet{Ramezanian24}}
\label{app:Ramezanian}

In this appendix, we provide a detailed comparison between our work and that of \citet{Ramezanian24}.

First, we show that our WF mechanism does not belong to \citeauthor{Ramezanian24}'s class of URC mechanisms.
To see this, consider a profile with $n = 3$ agents who all prefer the degenerate distribution $(1,0,0)$.
On this profile, URC mechanisms return
\begin{align}
\begin{pmatrix}
1/3&2/3  &0\\
1/3&1/3&1/3\\
1/3&0&2/3
\end{pmatrix}
\end{align}
or one of its row-permutations, due to the sequential nature of their second stage.
On the other hand, WF returns the uniform matrix 
\begin{align}
\begin{pmatrix}
1/3&1/3  &1/3\\
1/3&1/3&1/3\\
1/3&1/3&1/3
\end{pmatrix}.
\end{align}
Intuitively, the output of WF appears to be more appropriate in light of the symmetry among agents. 
More formally, URC fails equal treatment of equals, which WF satisfies.

Next, we demonstrate that our results answer three open questions posed by \citet{Ramezanian24}.\footnote{For definitions of properties not defined in our paper, we refer to \citeauthor{Ramezanian24}'s work.}
\begin{itemize}
\item \citet[pp.~27--28]{Ramezanian24} asked whether there exists a mechanism that is strategyproof, Pareto efficient, and satisfies either envy-freeness or anonymity, yet is not welfare-equivalent to URC mechanisms. 
Our QP mechanism provides such an example---note that QP satisfies both anonymity and envy-freeness. The fact that QP is not welfare-equivalent to URC can be seen from Examples 2 and 3 or Section 5 (note that WF is welfare-equivalent to URC).
\item \citet[p.~26]{Ramezanian24} asked whether there exists a mechanism that is strategyproof and Pareto efficient but does not satisfy the in-betweenness property.
Our WF mechanism provides such an example. Specifically, consider a profile~$P$ with $n = 3$ where the first two agents prefer $(0,0,1)$ and the third agent prefers $(0,1/2,1/2)$.
On this profile, WF returns 
\begin{align}
\mechWF(P)=\begin{pmatrix}
11/24&5/24  &1/3\\
11/24&5/24  &1/3\\
1/12&7/12&1/3
\end{pmatrix}.
\end{align}
On the other hand, consider the profile~$P'$ where the first agent instead prefers $(1/3, 1/6, 1/2)$, which is ``between'' $(0,0,1)$ and $(11/24,5/24,1/3)$.
On this profile, WF returns 
\begin{align}
\mechWF(P')=\begin{pmatrix}
5/12&1/4  &1/3\\
1/2&1/6  &1/3\\
1/12&7/12&1/3
\end{pmatrix}.
\end{align}
Since $5/12 < 11/24$, the first agent receives less of the first object even though she demands more.
\item \citet[p.~26]{Ramezanian24} asked whether there exists a mechanism that is strategyproof and Pareto efficient but not non-bossy.\footnote{Note that Ramezanian’s non-bossiness notion is somewhat unusual, as it only considers over-demanded objects.}
Our QP mechanism again provides such an example. 
Specifically, consider a profile~$P$ with $n = 3$ where the three agents prefer $(0, 0.6, 0.4)$, $(0.6, 0.4, 0)$, and $(0.2, 0.2, 0.6)$, respectively.
On this profile, QP returns
\begin{align}
\mechQP(P)=\begin{pmatrix}
0.2&0.4  &0.4\\
0.6&0.4  &0\\
0.2&0.2& 0.6
\end{pmatrix}.
\end{align}
On the other hand, consider the profile~$P'$ where the third agent instead prefers $(0.4, 0.2, 0.4)$.
On this profile, QP returns
\begin{align}
\mechQP(P')=\begin{pmatrix}
0&0.55  &0.45\\
0.6&0.25  &0.15\\
0.4&0.2& 0.4
\end{pmatrix}.
\end{align}
Hence, although the third agent receives the same amount of over-demanded object (i.e., the second object) as before, the other agents receive different amounts of this object than before.
\end{itemize}

\section{The QP Mechanism Is Not Weakly Group-Strategyproof}
\label{sec:weak-gsp}

In this appendix, we show that the QP mechanism fails weak group-strategyproofness (and therefore group-strategyproofness).

\begin{proposition}
\label{prop:weak-gsp-counterexample}
The mechanism $\mechQP$ is not weakly group-strategyproof when $n=4$.
\end{proposition}
\begin{proof}
Let $N=M=\{1,2,3,4\}$, and consider the profile $P\in\Delta^N$ given by
\begin{align*}
P\coloneqq
\frac1{15}
\begin{pmatrix}
2 & 2 & 2 & 9\\
2 & 4 & 6 & 3\\
0 & 8 & 3 & 4\\
6 & 6 & 1 & 2
\end{pmatrix}.
\end{align*}
Let $S\coloneqq\{1,4\}$, and define the manipulated profile $Q\in\Delta^N$ by
\begin{align*}
Q\coloneqq
\begin{pmatrix}
0 & \frac14 & \frac14 & \frac12\\
\frac{2}{15} & \frac{4}{15} & \frac25 & \frac15\\
0 & \frac{8}{15} & \frac15 & \frac{4}{15}\\
0 & 1 & 0 & 0
\end{pmatrix}.
\end{align*}
Thus, $Q=(\bm{q}_S,\bm{p}_{-S})$ with $\bm{q}_1=(0,\frac14,\frac14,\frac12)$ and $\bm{q}_4=(0,1,0,0)$.
Let $X\coloneqq\mechQP(P)$ and $X'\coloneqq\mechQP(Q)$.
Then,
\begin{align*}
X=
\frac{1}{75}
\begin{pmatrix}
18 & 10 & 17 & 30\\
11 & 19 & 30 & 15\\
16 & 24 & 15 & 20\\
30 & 22 & 13 & 10
\end{pmatrix},
\qquad
X'=
\frac{1}{120}
\begin{pmatrix}
15 & 15 & 30 & 60\\
24 & 24 & 48 & 24\\
32 & 32 & 24 & 32\\
49 & 49 & 18 & 4
\end{pmatrix}.
\end{align*}

We certify that $X=\mechQP(P)$ and $X'=\mechQP(Q)$ by exhibiting KKT certificates for problem $({\rm QP})$.
For $P$, the column sums are $(\frac23,\frac43,\frac45,\frac65)$, and therefore
$M_<=\{1,3\}$ and $M_>=\{2,4\}$.
Let
\begin{align*}
(\alpha_1,\alpha_2,\alpha_3,\alpha_4)&\coloneqq\left(\frac{18}{75},\frac{11}{75},\frac{16}{75},\frac{14}{75}\right),\\
(\beta_1,\beta_2,\beta_3,\beta_4)&\coloneqq\left(0,\frac{8}{75},-\frac{1}{75},\frac{12}{75}\right),
\end{align*}
and define $X$ from $\alpha,\beta$ via \Cref{prop:opt-rep}.
Then, the displayed matrix $X$ is obtained.
Define multipliers by
\begin{align*}
\lambda_{ij}&\coloneqq \max\{\alpha_i+\beta_j-p_{ij},\,0\}\qquad (j\in M_\ge),\\
\mu_{ij}&\coloneqq \max\{p_{ij}-\alpha_i-\beta_j,\,0\}\qquad (j\in M_\le),
\end{align*}
and set $\lambda_{ij}=0$ for $j\notin M_\ge$ and $\mu_{ij}=0$ for $j\notin M_\le$.
Then, $(X,\alpha,\beta,\lambda,\mu)$ satisfies the KKT conditions (a)--(h) for $({\rm QP})$ at $P$.
Since $({\rm QP})$ is strictly convex, $X$ is its unique optimizer, and therefore $X=\mechQP(P)$.

For $Q$, the column sums are $(\frac{2}{15},\frac{41}{20},\frac{17}{20},\frac{29}{30})$, and therefore
$M_<=\{1,3,4\}$ and $M_>=\{2\}$.
Let
\begin{align*}
(\alpha_1,\alpha_2,\alpha_3,\alpha_4)&\coloneqq\left(\frac{15}{120},\frac{24}{120},\frac{32}{120},\frac{49}{120}\right),\\
(\beta_1,\beta_2,\beta_3,\beta_4)&\coloneqq\left(0,0,-\frac{31}{120},-\frac{45}{120}\right),
\end{align*}
define $X'$ from $\alpha,\beta$ via \Cref{prop:opt-rep}, and define $\lambda,\mu$ as above.
Then, the displayed matrix $X'$ is obtained, and the same KKT verification shows that $X'=\mechQP(Q)$.

Finally, observe that
\begin{align*}
\du(\bm{x}_1,\bm{p}_1)=\frac25,
\qquad
\du(\bm{x}_4,\bm{p}_4)=\frac{16}{75},
\qquad
\du(\bm{x}'_1,\bm{p}_1)=\frac{7}{30},
\qquad
\du(\bm{x}'_4,\bm{p}_4)=\frac15.
\end{align*}
Thus, both members of $S$ strictly improve through their manipulation, thereby violating weak group-strategyproofness.
\end{proof}

\section{Proof of \Cref{prop:qp-sp}}
\label{app:prop-qp-sp}

Fix a profile $P\in\Delta^N$ and an agent $i\in N$.
For a report $\bm{q}_i\in\Delta$, define the loss
\begin{align*}
F_{P,i}(\bm{q}_i)\coloneqq\du\!\left(\mechQP(\bm{q}_i,\bm{p}_{-i})_i,\,\bm{p}_i\right).
\end{align*}
It suffices to show that $F_{P,i}(\bm{q}_i)\ge F_{P,i}(\bm{p}_i)$ for every report $\bm{q}_i\in\Delta$.
The claim is immediate when $\bm{q}_i=\bm{p}_i$, so we fix $\bm{q}_i\in\Delta\setminus\{\bm{p}_i\}$ from now on.

Set $\bm{d}\coloneqq \bm{q}_i-\bm{p}_i$, and note that $\bm{p}_i+\tau \bm{d}\in\Delta$ for all $\tau\in[0,1]$.
For $\tau\in[0,1]$, define
\begin{align*}
\bm{q}_i(\tau)&\coloneqq \bm{p}_i+\tau \bm{d},\\
X(\tau)&\coloneqq \mechQP(\bm{q}_i(\tau),\bm{p}_{-i}),\\
f(\tau)&\coloneqq F_{P,i}(\bm{q}_i(\tau))=\du\!\left(\bm{x}_i(\tau),\,\bm{p}_i\right).
\end{align*}
Our goal is to prove the endpoint inequality $f(1)\ge f(0)$.

\subsection{Ray structure}
\label{subsec:sp-proof-a1-a2}

The first step is to reduce the manipulation analysis to finitely many affine intervals along the
line segment from the truthful report to the deviating report.
\begin{lemma}
\label{lem:sp-ray-structure}
There exist an integer $K\ge0$ and breakpoints
\begin{align*}
0=\tau_0<\tau_1<\cdots<\tau_K<\tau_{K+1}=1
\end{align*}
such that for each half-open interval $I_\ell\coloneqq[\tau_\ell,\tau_{\ell+1})$ with $\ell\in\{0,1,\dots,K\}$, letting $\bar\tau_\ell\coloneqq(\tau_\ell+\tau_{\ell+1})/2$, the following hold:
\begin{itemize}
\item[(i)] For every $r\in N$ and $j\in M$, there exist $x_{rj}^{(\ell,1)},x_{rj}^{(\ell,0)}\in\mathbb R$ such that
\begin{align}
x_{rj}(\tau)=x_{rj}^{(\ell,1)}\cdot(\tau-\bar\tau_\ell)+x_{rj}^{(\ell,0)}
\qquad(\tau\in I_\ell). 
\label{eq:sp-ray-structure-1}
\end{align}
\item[(ii)] There exist $a^{(\ell)},b^{(\ell)}\in\mathbb R$ such that
\begin{align*}
f(\tau)=a^{(\ell)}\cdot(\tau-\bar\tau_\ell)+b^{(\ell)}
\qquad(\tau\in I_\ell).
\end{align*}
\item[(iii)] The map $X(\tau)$ is continuous on $[0,1]$.
\item[(iv)] On $I_\ell$, the direction of every coordinate inequality in~{\rm(QP)} (i.e., inequalities either upper-bounding or lower-bounding $x_{rj}(\tau)$) is fixed, and
each corresponding slack is either strictly positive throughout
$(\tau_\ell,\tau_{\ell+1})$ or identically zero on $I_\ell$.
\end{itemize}
\end{lemma}

\begin{proof}
Consider the reported profile $Q(\tau)\in\Delta^N$ given by $\bm{q}_i(\tau)=\bm{p}_i+\tau \bm{d}$ and
$\bm{q}_r(\tau)=\bm{p}_r$ for $r\neq i$.
For each column $j\in M$, the quantity
$s_j(\tau)\coloneqq \sum_{r\in N}q_{rj}(\tau)-1$ is affine in $\tau$, so it changes sign at most once on $[0,1]$.
Define $\kappa_j\in[0,1]$ as follows: if $d_j\neq 0$ and there exists $\tau\in[0,1]$ such that $\sum_{r\in N}q_{rj}(\tau)=1$, let $\kappa_j$ be that unique $\tau$; otherwise set $\kappa_j\coloneqq 1$.
Let $0=t_0<t_1<\cdots<t_L=1$ be the sorted list of distinct values in $\{0,1\}\cup\{\kappa_j\mid j\in M\}$.
Then, on each interval $[t_\ell,t_{\ell+1})$, for every $j\in M$, either $\sum_{r\in N}q_{rj}(\tau)\ge1$ holds for all $\tau\in[t_\ell,t_{\ell+1})$ or $\sum_{r\in N}q_{rj}(\tau)\le1$ holds for all $\tau\in[t_\ell,t_{\ell+1})$.
At points where $\sum_{r\in N}q_{rj}(\tau)=1$, the inequality direction does not matter: combined with the column-sum constraint $\sum_{r\in N}x_{rj}=1$, it forces $x_{rj}=q_{rj}(\tau)$ for every $r\in N$.
Therefore, on each interval $[t_\ell,t_{\ell+1})$, we may fix the inequality direction for each column $j$, and the feasible set of~{\rm(QP)} becomes polyhedral with affine dependence on $\tau$.

On each such interval, the mechanism outcome $X(\tau)$ is the unique optimizer of a strictly convex quadratic program with polyhedral constraints and affine dependence on $\tau$.
Standard results on parametric strictly convex quadratic programs imply that the optimizer map $X(\tau)$ is continuous and piecewise affine, and admits a finite refinement into critical regions on which $X(\tau)$ is affine (see \citet[Ch.~4]{BonnansShapiro00} or \citet[Ch.~6]{BorrelliBemporadMorari17}).
Collecting the critical-region refinements over all the intervals yields a finite partition $\{I'_\ell\}_{\ell=0}^{K'}$ of $[0,1]$ on which each coordinate $x_{rj}(\tau)$ is affine, proving (i).
Since $X(\tau)$ is continuous on $[0,1]$, (iii) is also established.

For each $j\in M$, the difference $x_{ij}(\tau)-p_{ij}$ is affine on each $I'_\ell$, so it changes sign at most once on $I'_\ell$.
For each coordinate inequality in~{\rm(QP)}, define its slack on an interval with fixed inequality
direction by $q_{rj}(\tau)-x_{rj}(\tau)$ for an upper constraint and by
$x_{rj}(\tau)-q_{rj}(\tau)$ for a lower constraint. Each such slack is affine and nonnegative on
the interval. We refine the partition at the finitely many zeros of the nonzero affine slack
functions and at the finitely many points where $x_{ij}(\tau)=p_{ij}$. After removing the point
$1$ and intervals of length $0$, we obtain a subpartition
$\{[\tau_\ell,\tau_{\ell+1})\}_{\ell\in\{0,\dots,K\}}$ of $[0,1)$ where $\tau_0=0$.
On each refined interval $I_\ell \coloneqq [\tau_\ell,\tau_{\ell+1})$, for all $\tau\in(\tau_\ell,\tau_{\ell+1})$, the value of $x_{ij}(\tau)-p_{ij}$ is either always positive, always negative, or zero.
On such an interval, we have $f(\tau)=\du(\bm{x}_i(\tau),\bm{p}_i)=\sum_{j\in M}|x_{ij}(\tau)-p_{ij}|$,
and each absolute value is an affine function of $\tau$ with fixed sign.
Hence, $f(\tau)$ is affine on each interval, proving (ii).
The same refinement proves~(iv), and (i) and (iii) are inherited from the critical-region
partition.
\end{proof}

\subsection{Endpoint inequality from interval slopes}
\label{subsec:sp-proof-a3}

In this subsection, we use the terminology from \Cref{lem:sp-ray-structure}.
For $\ell\in\{0,1,\dots,K\}$, let $I_\ell = [\tau_\ell, \tau_{\ell+1})$ be the interval indicated by \Cref{lem:sp-ray-structure}, and let $\bar\tau_\ell\coloneqq(\tau_\ell+\tau_{\ell+1})/2$.
The function $f$ is continuous on $[0,1]$ and piecewise affine on $[0,1)$.
By continuity, we also have $f(1) = a^{(K)}\cdot(1-\bar\tau_{K})+b^{(K)}$, and 
it follows that $f(\tau_{\ell+1}) - f(\tau_{\ell}) = a^{(\ell)}\cdot(\tau_{\ell+1}-\tau_{\ell})$ for each $\ell$.
The telescoping sum implies 
\begin{align*}
f(1)-f(0)=\sum_{\ell=0}^{K}a^{(\ell)}\cdot(\tau_{\ell+1}-\tau_\ell).
\end{align*}
Hence, it suffices to show that $a^{(\ell)}\ge0$ for every $\ell$.

Fix an interval $I_\ell=[\tau_\ell,\tau_{\ell+1})$.
By construction of the sign blocks in the proof of \Cref{lem:sp-ray-structure}, for each column $j\in M$, the constraints in~{\rm(QP)} include either $x_{rj}\le q_{rj}(\tau)$ for all $r\in N$ and $\tau\in I_\ell$, or $x_{rj}\ge q_{rj}(\tau)$ for all $r\in N$ and $\tau\in I_\ell$.
The corresponding slack is affine and nonnegative on~$I_\ell$. If such a slack is zero at the
interior point $\bar\tau_\ell$, then affine nonnegativity forces the slack to be identically zero
on $I_\ell$. Otherwise, the slack is positive at $\bar\tau_\ell$.

Define the ``frozen set''
\begin{align*}
\mathcal F^{(\ell)}\coloneqq\{(r,j)\in N\times M\mid x_{rj}(\tau)=q_{rj}(\tau)\ \text{for all }\tau\in I_\ell\}.
\end{align*}
Since $X(\tau)$ is affine on $I_\ell$, it is differentiable on $(\tau_\ell,\tau_{\ell+1})$ with a constant derivative, which equals the slope matrix $X^{(\ell,1)}\coloneqq(x_{rj}^{(\ell,1)})_{r\in N,\,j\in M}$.
Moreover, every inequality constraint in~{\rm(QP)} that is tight at $\tau=\bar\tau_\ell$ is encoded by $\mathcal F^{(\ell)}$, and every inequality constraint on $(N\times M)\setminus\mathcal F^{(\ell)}$ is slack at $\tau=\bar\tau_\ell$.
We see that the slope matrix $X^{(\ell,1)}$ is characterized as the unique minimizer of the tangent program below:
\begin{alignat}{2}
({\rm T}_\ell)\quad
\min\ \ &\textstyle\frac{1}{2}\sum_{r\in N}\sum_{j\in M}y_{rj}^2\\
\text{s.t.}\ \ 
&\textstyle\sum_{j\in M}y_{rj}=0  &\quad& (r\in N),\notag\\
&\textstyle\sum_{r\in N}y_{rj}=0  &     & (j\in M),\notag\\
&y_{rj}=d_j             &     & ((r,j)\in\mathcal{F}^{(\ell)},\ r=i),\notag\\
&y_{rj}=0               &     & ((r,j)\in\mathcal{F}^{(\ell)},\ r\neq i).\notag
\end{alignat}
Since the feasible set is a nonempty affine set and the objective is strictly convex, the minimizer is unique. 

We now justify the tangent program $({\rm T}_\ell)$. 
Let $\tau\in(\tau_\ell,\tau_{\ell+1})$.
The active Lagrangian of (QP) is
\begin{align*}
\frac{1}{2}\sum_{r\in N}\sum_{j\in M}x_{rj}^2
+\sum_{r\in N}\alpha_r\left(1-\sum_{j\in M}x_{rj}\right)
+\sum_{j\in M}\beta_j\left(1-\sum_{r\in N}x_{rj}\right)
+\sum_{(r,j)\in\mathcal F^{(\ell)}}\nu_{rj}
\bigl(x_{rj}-q_{rj}(\tau)\bigr).
\end{align*}
Here, the term $x_{rj}-q_{rj}(\tau)$ for $(r,j)\notin \mathcal{F}^{(\ell)}$ does not appear because omitting the corresponding constraint does not change the optimal solution.
Hence, the stationarity condition is
\begin{align}
x_{rj}(\tau)-\alpha_r(\tau)-\beta_j(\tau)
+\bm{1}_{(r,j)\in\mathcal{F}^{(\ell)}}\nu_{rj}(\tau)=0\quad
(r\in N,\,j\in M).
\label{eq:lagrangian_tau}
\end{align}
Take two distinct interior points $\tau',\tau''\in(\tau_\ell,\tau_{\ell+1})$. 
Since $X(\tau)$ is affine on $I_\ell$ by \Cref{lem:sp-ray-structure}, subtracting \eqref{eq:lagrangian_tau} at $\tau'$ and $\tau''$ 
and dividing by $\tau'-\tau''$ gives
\begin{align*}
x^{(\ell,1)}_{rj}
-\widehat\alpha_r-\widehat\beta_j
+\bm{1}_{(r,j)\in\mathcal{F}^{(\ell)}}\widehat\nu_{rj}=0\quad
(r\in N,\,j\in M),
\end{align*}
where the variables with hats are the corresponding difference quotients of the multipliers. 
Similarly, subtracting the active feasibility constraints gives
\begin{align*}
\sum_{j\in M}x^{(\ell,1)}_{rj}=0 \quad (r\in N),\qquad
\sum_{r\in N}x^{(\ell,1)}_{rj}=0 \quad (j\in M).
\end{align*}
For each $(r,j) \in \mathcal{F}^{(\ell)}$, since $x_{rj}(\tau)=q_{rj}(\tau) = p_{rj}+\tau d_{j} \cdot \bm{1}_{r=i}$ for all $\tau \in I_\ell$, we have
\begin{align*}
x^{(\ell,1)}_{rj}=
\begin{cases}
d_j & ((r,j)\in\mathcal F^{(\ell)},\, r=i),\\
0 & ((r,j)\in\mathcal F^{(\ell)},\, r\ne i).
\end{cases}
\end{align*}
These are exactly the KKT conditions for $({\rm T}_\ell)$, with
$(\widehat\alpha_{r})_{r\in N},(\widehat\beta_{j})_{j\in M},(\widehat\nu_{rj})_{r\in N,\,j\in M}$ as its multipliers. 
Therefore, $X^{(\ell,1)}$ is the unique optimizer of $({\rm T}_\ell)$.

\paragraph{Intervalwise slope nonnegativity.}
\label{subsubsec:sp-a3-interval-nonnegativity}

It remains to show that every interval slope in the piecewise-affine loss is nonnegative.
\begin{lemma}
\label{lem:sp-interval-slope-nonnegativity}
For every $\ell\in\{0,\dots,K\}$, it holds that $a^{(\ell)}\ge0$, where $a^{(\ell)}$ is defined as in \Cref{lem:sp-ray-structure}.
\end{lemma}

\begin{proof}
Fix $\ell\in\{0,\dots,K\}$.
Define
\begin{align*}
M_+^{(\ell)}\coloneqq\{j\in M\mid x_{ij}^{(\ell,0)}-p_{ij}>0\}
\quad\text{and}\quad
M_-^{(\ell)}\coloneqq\{j\in M\mid  x_{ij}^{(\ell,0)}-p_{ij}<0\}.
\end{align*}
By construction of the partition in \Cref{lem:sp-ray-structure}, 
each affine function $g_j(\tau)\coloneqq x_{ij}(\tau)-p_{ij}$ is either 
strictly positive, strictly negative, or identically zero on $(\tau_\ell,\tau_{\ell+1})$.
Since $g_j(\bar{\tau}_\ell)=x_{ij}^{(\ell,0)}-p_{ij}$, 
the affine function $g_j(\tau)$ has the following sign behavior on the interval $(\tau_\ell,\tau_{\ell+1})$:
it is strictly positive if $j\in M_+^{(\ell)}$, 
strictly negative if $j\in M_-^{(\ell)}$, and 
identically zero otherwise.
Therefore, on this interval, 
\begin{align*}
a^{(\ell)}
&=\dv{\tau}\du(\bm{x}_i(\tau),\bm{p}_i)
=\dv{\tau}\sum_{j\in M}|x_{ij}(\tau)-p_{ij}|\\
&
=\sum_{j\in M_+^{(\ell)}}\dv{\tau} (x_{ij}(\tau)-p_{ij})
+\sum_{j\in M_-^{(\ell)}}\dv{\tau} (p_{ij}-x_{ij}(\tau))
=\sum_{j\in M_+^{(\ell)}}x_{ij}^{(\ell,1)}
-\sum_{j\in M_-^{(\ell)}}x_{ij}^{(\ell,1)}.
\end{align*}

Assume for contradiction that $a^{(\ell)}<0$.
Since $x_{ij}^{(\ell,1)}=0$ for $j\notin M_+^{(\ell)}\cup M_-^{(\ell)}$, the row-sum equality in $({\rm T}_\ell)$ implies that
$\sum_{j\in M_+^{(\ell)}}x_{ij}^{(\ell,1)}+\sum_{j\in M_-^{(\ell)}}x_{ij}^{(\ell,1)}=0$.
Hence, we have
\begin{align*}
\sum_{j\in M_+^{(\ell)}}x_{ij}^{(\ell,1)}=\frac12 a^{(\ell)}<0
\quad\text{and}\quad
\sum_{j\in M_-^{(\ell)}}x_{ij}^{(\ell,1)}=-\frac12 a^{(\ell)}>0.
\end{align*}

We next view $X^{(\ell,1)}$ as a circulation in the directed graph with vertex set $N\cup M$, defined as follows.
For each $(r,j)\in N\times M$ with $x_{rj}^{(\ell,1)}\neq 0$, introduce an arc oriented
$r\to j$ if $x_{rj}^{(\ell,1)}>0$ and $j\to r$ if $x_{rj}^{(\ell,1)}<0$, and assign it the
(positive) value $w_{rj}\coloneqq |x_{rj}^{(\ell,1)}|$.
The row- and column-sum equalities in $({\rm T}_\ell)$ imply flow conservation at
every vertex of this directed graph, and therefore $W$ is a circulation.
By the standard circulation decomposition theorem, $W$ can be written as a positive weighted sum of directed cycles
(see, e.g., \citet[Property~3.6]{AhujaMaOr93}).
Take any decomposition of the circulation~$W$ into simple directed cycles.
Since $\sum_{j\in M_+^{(\ell)}}x_{ij}^{(\ell,1)}<0$ and $\sum_{j\in M_-^{(\ell)}}x_{ij}^{(\ell,1)}>0$,
at least one cycle in the decomposition enters $i$ from some $j^+\in M_+^{(\ell)}$ and leaves~$i$ to some $j^-\in M_-^{(\ell)}$.
Fix such a directed cycle, and let $C\subseteq N\times M$ be the set of agent--object pairs $(r,j)$ such that $(r,j)$ or $(j,r)$ appears in the cycle.

We claim that $C\subseteq (N\times M)\setminus\mathcal F^{(\ell)}$.
For $(r,j)\in\mathcal F^{(\ell)}$ with $r\neq i$, no such edge can belong to~$C$, since $({\rm T}_\ell)$ forces $x_{rj}^{(\ell,1)}=0$ whereas every edge of $C$ has $x_{rj}^{(\ell,1)}\neq 0$.
Thus, it remains to consider edges incident to $i$.
Suppose $(i,j)\in\mathcal F^{(\ell)}$.
Then, $({\rm T}_\ell)$ forces $x_{ij}^{(\ell,1)}=d_j$. 
On the other hand, since $x_{ij}(\bar\tau_\ell)=x_{ij}^{(\ell,0)}$ by \eqref{eq:sp-ray-structure-1} and $x_{ij}(\bar\tau_\ell)=q_{ij}(\bar\tau_\ell)=p_{ij}+\bar\tau_\ell d_j$ by $(i,j)\in\mathcal F^{(\ell)}$, we obtain $x_{ij}^{(\ell,0)}-p_{ij}=\bar\tau_\ell d_j$.
Since $\bar\tau_\ell>0$, the quantities $x_{ij}^{(\ell,1)} = d_j$ and $x_{ij}^{(\ell,0)}-p_{ij}$ have the same sign.
If $x_{ij}^{(\ell,0)}-p_{ij}> 0$, then $j \in M_+^{(\ell)}$ and $j\neq j^-$, but we also have $j \neq j^+$ since $x_{ij}^{(\ell,1)} > 0$ while $x_{ij^+}^{(\ell,1)}<0$ (due to the arc $j^+ \to i$).
Similarly, if $x_{ij}^{(\ell,0)}-p_{ij} < 0$, then $j \not\in \{j^+, j^-\}$. 
Hence, $(i,j^-)\notin\mathcal F^{(\ell)}$ and $(i,j^+)\notin\mathcal F^{(\ell)}$, and therefore
$C\subseteq (N\times M)\setminus\mathcal F^{(\ell)}$.

Define the incidence matrix $V\in\{-1,0,1\}^{N\times M}$ of $C$ by
\begin{align*}
v_{rj}\coloneqq
\begin{cases}
1  & \text{if }(r,j)\in C\text{ and }(r,j)\text{ is oriented }r\to j,\\
-1 & \text{if }(r,j)\in C\text{ and }(r,j)\text{ is oriented }j\to r,\\
0  & \text{if }(r,j)\notin C.
\end{cases}
\end{align*}
Then, we have $\sum_{j\in M}v_{rj}=0$ for $r\in N$, $\sum_{r\in N}v_{rj}=0$ for $j\in M$, and
$v_{rj}x_{rj}^{(\ell,1)}=|x_{rj}^{(\ell,1)}|$ for $(r,j)\in C$.
In particular, $\sum_{r\in N}\sum_{j\in M}x_{rj}^{(\ell,1)}v_{rj}=\sum_{(r,j)\in C}|x_{rj}^{(\ell,1)}|>0$.

Define $\varepsilon\coloneqq \min\bigl\{|x_{rj}^{(\ell,1)}| \,\big|\, (r,j)\in C\bigr\}$,
which is strictly positive.
Set $Y\coloneqq X^{(\ell,1)}-\varepsilon V$.
Because $C\subseteq(N\times M)\setminus\mathcal F^{(\ell)}$, no frozen equality constraint in
$({\rm T}_\ell)$ applies to edges of $C$, and hence $y_{rj}=x_{rj}^{(\ell,1)}$
holds for every $(r,j)\in\mathcal F^{(\ell)}$.
Since $C$ is a directed cycle, $V$ satisfies the row- and column-sum equalities in
$({\rm T}_\ell)$, and therefore $Y$ is feasible in $({\rm T}_\ell)$.

We now compare objective values in $({\rm T}_\ell)$.
Since $v_{rj}=0$ outside $C$, we have
\begin{align*}
\frac12\sum_{r\in N}\sum_{j\in M}y_{rj}^2
&=\frac12\sum_{r\in N}\sum_{j\in M}\left(x_{rj}^{(\ell,1)}-\varepsilon v_{rj}\right)^2
=\frac12\sum_{r\in N}\sum_{j\in M}\left(x_{rj}^{(\ell,1)}\right)^2
-\varepsilon\sum_{(r,j)\in C}|x_{rj}^{(\ell,1)}|
+\frac{\varepsilon^2}{2}\sum_{(r,j)\in C}v_{rj}^2.
\end{align*}
Because $v_{rj}^2=1$ for $(r,j)\in C$, the last sum equals $|C|$, the number of edges of $C$.
Moreover, by the definition of $\varepsilon$, we have
$\sum_{(r,j)\in C}|x_{rj}^{(\ell,1)}|\ge \varepsilon|C|$.
Consequently,
\begin{align*}
\frac12\sum_{r\in N}\sum_{j\in M}y_{rj}^2
\le
\frac12\sum_{r\in N}\sum_{j\in M}\left(x_{rj}^{(\ell,1)}\right)^2
-\frac{\varepsilon^2}{2}|C|
<
\frac12\sum_{r\in N}\sum_{j\in M}\left(x_{rj}^{(\ell,1)}\right)^2.
\end{align*}
This contradicts the optimality of $X^{(\ell,1)}$ in $({\rm T}_\ell)$.
It follows that $a^{(\ell)}\ge0$.
\end{proof}

Finally, we complete the proof of \Cref{prop:qp-sp}.

\begin{proof}[Proof of \Cref{prop:qp-sp}]
Fix $(P,i)$ and a report $\bm{q}_i\in\Delta$, and set $\bm{d}=\bm{q}_i-\bm{p}_i$.
If $\bm{d}=\bm{0}$, then $f(1)=f(0)$.
Else, by \Cref{lem:sp-ray-structure}, $f$ is continuous and piecewise affine on $[0,1]$.
Hence,
\begin{align*}
f(1)-f(0)=\sum_{\ell=0}^{K}a^{(\ell)}(\tau_{\ell+1}-\tau_\ell).
\end{align*}
By \Cref{lem:sp-interval-slope-nonnegativity}, each $a^{(\ell)}$ is nonnegative, and each interval length $\tau_{\ell+1}-\tau_\ell$ is also nonnegative.
Therefore, $f(1)\ge f(0)$, which means that $F_{P,i}(\bm{q}_i)\ge F_{P,i}(\bm{p}_i)$, as desired.
\end{proof}

\end{document}